\documentclass[a4paper,UKenglish]{lipics}
\usepackage{tikz}
\usetikzlibrary{shapes,arrows,shadows,positioning,matrix,patterns,decorations.markings}
\usepackage{tkz-euclide}
\usetikzlibrary{calc,intersections}
\usetkzobj{all}
\usepackage{amsmath}
\usepackage{graphicx}
\usepackage{ifthen}
\usepackage{amssymb}
\usepackage{lineno}
\usepackage{enumerate}
\usepackage{comment}
\usepackage{fancyhdr}

\newenvironment{proofof}[1][Proof of]{\par
  \pushQED{\qed}%
  \normalfont 
  \trivlist
  \item[\hskip\labelsep
        \color{darkgray}\sffamily\bfseries
    #1]\ignorespaces
}{%
  \popQED\endtrivlist%\@endpefalse
}

\newcommand{\red}{black}
\newcommand{\purple}{black}
\newcommand{\pathf}[1]{\ensuremath{\mathcal{#1}}}
\newcommand{\pathR}{{\color{\purple} \ensuremath{\pathf{R}}}}
\newcommand{\pathP}{{\color{\red} \ensuremath{\mathcal{P}}}}
\newcommand{\subP}[1]{\mbox{\textless}#1\mbox{\textgreater}}
\newcommand{\arcP}[2]{{\color{\red} \ensuremath{\pathf{A}_{#1}\subP{#2}}}}
\newcommand{\dR}[1]{{\color{\purple} \ensuremath{\pathR\subP{#1}}}}
\newcommand{\dP}[1]{{\color{\red}\ensuremath{\pathP_{#1}}}}

\newcommand{\Cpi}[1]{{\color{\red}\ensuremath{C'_{#1}}}}

\newcommand{\snailP}[1]{\ensuremath{\pathf{S}_{#1}}}
\newcommand{\nodeList}{{0/250/p0/$s$/above},{338/250/t/$t$/below right},{54/277/p1/$p_1$/above left},{135/238/p2/$p_2$/},{216/290/p3/$p_3$/above},{257/297/p4/$p_4$/above left},{131/189/q1p/$q'_1$/below},{-10/225/q1/$q_1$/below left},{104/349/q2/$q_2$/above},{157/169/q3/$q_3$/below},{281/181/q4/$q_4$/below},{315/233/q5/$q_5$/below right}}

\newcommand{\routingPath}{{p0/p1/q1/1/below/90/0/1},{p1/p2/q2/0/above/-110/1/1},{p2/p3/q3/1/below/90/2/1},{p3/p4/q4/1/right/210/3/0},{p4/t/q5/1/above/-60/4/1}}

\newcommand{\otherTriangles}{{q1p/p1/q1},{q1p/p1/p2},{p4/q4/q5},{q3/q4/p3},{p2/q1p/q3},{p2/q2/p3}}
\newcommand{\clipExample}{\clip (-0.3,1.3) rectangle (3.8,3.75);}
\newcommand{\yellow}{gray}
\newcommand{\blue}{black}
\tikzstyle{path}=[color=\purple,line width=2pt, opacity=.6]
\tikzstyle{pathBound}=[color=brown,line width=2pt]
\tikzstyle{emptyCircle}=[color=gray]
\tikzstyle{triangles}=[color=\blue]
\tikzstyle{dP}=[color=\red]
\tikzstyle{dPCurve}=[color=\red, line width=2pt]
\tikzstyle{dPpCurve}=[dPCurve,style=dashdotted,line width=2pt]
\tikzstyle{LHS} = [color=black, line width=3pt,opacity=.6]
\tikzstyle{RHS} = [color=black, line width=3pt,opacity=.3]
\tikzstyle{Cprime} = [color=\red]
\newcommand{\deltaMin}{0.185043874}
\newcommand{\tkzComputeAllPoints}{
\tkzDefMidPoint(p0,p0)\tkzGetPoint{s}

\foreach \A/\B/\C  in \otherTriangles
{\tkzCircumCenter(\A,\B,\C)\tkzGetPoint{G\A\B\C} }

\foreach \A/\B/\C/\cw/\cpos/\cposangle/\cnum/\ctype  in \routingPath
  {    \tkzDefMidPoint(\A,\B)\tkzGetPoint{mid\A\B}
       \tkzDefLine[parallel=through \A](s,t)\tkzGetPoint{para\A}
       \tkzDefLine[orthogonal=through mid\A\B](\A,\B) \tkzGetPoint{ortho\A\B}
       \tkzInterLL(mid\A\B,ortho\A\B)(\A,para\A) \tkzGetPoint{Gp\A\B\C}
  \tkzCircumCenter(\A,\B,\C)\tkzGetPoint{G\A\B\C}
  \tkzDefMidPoint(\A,\A)\tkzGetPoint{primep\cnum}
}

\tkzDefPointWith[linear,K=4.75](midp3p4,Gpp3p4q4) \tkzGetPoint{Gpp3p4q4}
\tkzDefLine[parallel=through Gpp3p4q4](s,t)\tkzGetPoint{HGpp3p4q4}
\tkzInterLC(Gpp3p4q4,HGpp3p4q4)(Gpp3p4q4,p3)\tkzGetPoints{primep3}{toto}

\foreach \A/\B/\C/\cw/\cpos/\cposangle/\cnum/\ctype  in \routingPath
{ \tkzDefPointBy[projection=onto s--t](primep\cnum) \tkzGetPoint{sp\cnum}
\tkzDefPointBy[projection=onto s--t](primep\cnum) \tkzGetPoint{s\cnum}}

\tkzDefMidPoint(p1,p1)\tkzGetPoint{t0}
\tkzDefMidPoint(p2,p2)\tkzGetPoint{t1}
\tkzDefMidPoint(p3,p3)\tkzGetPoint{t2}
\tkzDefMidPoint(t,t)\tkzGetPoint{t3}
\tkzDefMidPoint(t,t)\tkzGetPoint{t4}
\tkzDefMidPoint(t,t)\tkzGetPoint{t5}
\tkzDefPointBy[translation= from s1 to p1](s) \tkzGetPoint{s0}
\tkzDefPointBy[projection=onto p1--s1](p2) \tkzGetPoint{s1}
\tkzDefPointBy[projection=onto p2--s2](p3) \tkzGetPoint{s2}
}

\newcommand{\defPointsGreenBlue}{
\tkzInterCC(Op0,b0)(Op1,b1) \tkzGetFirstPoint{p1}

\tkzDefPoint(-0.2, 0){s} 
\tkzDefPoint(6, 0){t} 

\foreach \i in {0,1} {
\tkzDefLine[parallel=through Op\i](s,t)\tkzGetPoint{HO\i}
\tkzInterLC(Op\i,HO\i)(Op\i,b\i)\tkzGetPoints{pp\i}{toto}
\tkzDefLine[orthogonal=through pp\i](s,t) \tkzGetPoint{vertpp\i}
\tkzInterLL(pp\i,vertpp\i)(s,t)\tkzGetPoint{sp\i}
\tkzDefLine[parallel=through t\i](s,t) \tkzGetPoint{horit\i}
\tkzInterLL(pp\i,vertpp\i)(horit\i,t\i)\tkzGetPoint{s\i}
}
\tkzDefPointBy[rotation= center sp1 angle -90](sp0)\tkzGetPoint{Os}
\tkzDefLine[parallel=through Os](s,t) \tkzGetPoint{horits}
\tkzInterLL(pp0,vertpp0)(horits,Os)\tkzGetPoint{pps}
\tkzDefLine[orthogonal=through pp1](s,t) \tkzGetPoint{vertpp1}
\tkzInterLL(pp1,vertpp1)(s,t)\tkzGetPoint{sp1}
\tkzDefLine[parallel=through t1](s,t) \tkzGetPoint{horit1}
\tkzDrawLine(s,t)
\tkzDefLine[parallel=through b0](s,t) \tkzGetPoint{horib0}
\tkzDefPointBy[projection=onto b0--horib0](s0) \tkzGetPoint{spp0}
\tkzDefPointBy[projection=onto b0--horib0](b1) \tkzGetPoint{bp1}
\tkzDefPointBy[projection=onto s--t](p1) \tkzGetPoint{tp1}
\tkzDefPointBy[projection=onto p1--tp1](t1) \tkzGetPoint{tpp1}
\tkzDefPointBy[projection=onto sp0--pp0](p1) \tkzGetPoint{tpp0}
\tkzDefPointBy[translation= from sp1 to p1](sp0)\tkzGetPoint{X}

\tkzDefPointBy[translation= from sp1 to p1](Os)\tkzGetPoint{Os2}
\tkzDefPointBy[translation= from sp1 to p1](pps)\tkzGetPoint{pps2}
\tkzDefPointBy[translation= from sp1 to sp0](Op1)\tkzGetPoint{Op12}
\tkzDefPointBy[translation= from sp1 to sp0](pp1)\tkzGetPoint{pp12}

\tkzDefPointBy[rotation= center bp1 angle -90](spp0)\tkzGetPoint{Os3}

\tkzDefLine[parallel=through Os3](s,t) \tkzGetPoint{horits3}
\tkzInterLL(pp0,vertpp0)(horits3,Os3)\tkzGetPoint{pps3}

\tkzInterLC(spp0,p1)(Os3,pps3)\tkzGetSecondPoint{Z}

\tkzCalcLength(p1,Z) \tkzGetLength{dpq}
\tkzCalcLength(spp0,p1) \tkzGetLength{dsppp}
\pgfmathparse{(\dpq)/(\dsppp)}
\edef\tkz@q{\pgfmathresult}%

 \tkzDefPointBy[homothety= center spp0 ratio \pgfmathresult](Op0)\tkzGetPoint{Op0tmp}
 \tkzDefPointBy[homothety= center spp0 ratio \pgfmathresult](pp0)\tkzGetPoint{pp0tmp}
 \tkzDefPointBy[translation= from spp0 to p1](Op0tmp)\tkzGetPoint{Op0Y}
 \tkzDefPointBy[translation= from spp0 to p1](pp0tmp)\tkzGetPoint{pp0Y}
}

\newcommand{\LHS}{
\tkzDrawArc[LHS](Op0,p1)(pp0)
\tkzDrawSegment[LHS](pp0,sp0)
}

\newcommand{\RHS}{
\tkzDrawArc[RHS](Op1,p1)(pp1)
\tkzDrawSegment[RHS](pp1,sp1)
\tkzDrawArc[RHS](Os,sp1)(pps)
\tkzDrawSegment[RHS](pps,sp0)}

\newcommand{\figexampleRouting}{
\begin{tikzpicture}[scale=2.6]
\clipExample
\foreach \x\y/\name/\etiq/\pos  in \nodeList
  {\tkzDefPoint(\x/100,\y/100){\name}}

\tkzComputeAllPoints

\tkzDrawLine(s,t)

\foreach \A/\B/\C  in \otherTriangles
{\tkzCircumCenter(\A,\B,\C)\tkzGetPoint{G\A\B\C}
   \tkzDrawPolygon[triangles, color=gray](\A,\B,\C)
   }

\foreach \A/\B/\C/\cw/\cpos/\cposangle/\cnum/\ctype  in \routingPath
  {\tkzCircumCenter(\A,\B,\C)\tkzGetPoint{G\A\B\C}
   \tkzDrawPolygon[triangles](\A,\B,\C)
   \tkzDrawCircle[emptyCircle](G\A\B\C,\A)
   \tkzLabelCircle[\cpos=4pt](G\A\B\C,\A)(\cposangle){$C_{\cnum}$}
   \ifthenelse{\equal{\cw}{1}}{\tkzDrawArc[path](G\A\B\C,\B)(\A)}{\tkzDrawArc[path](G\A\B\C,\A)(\B)}
   \tkzDrawSegment[triangles, line width=2pt](\B,\A)}
\foreach \x\y/\name/\etiq/\pos  in \nodeList
  { \tkzDrawPoint[size=2](\name)
    \tkzLabelPoint[\pos](\name){\etiq}
  }
\end{tikzpicture}
}

\newcommand{\figLBChew}{
\begin{tikzpicture}[scale=1.8]
\clip (-1.25,2) rectangle (2.5,-1.25);
\draw [name path=horizontal,dashed] (-1,0) -- (10,0);
\draw [name path=C1,color=gray!50,radius=1cm] (0,0) circle;
\node at (-1, 0) [fill, inner sep = 1pt, circle, label=180:$s$] (s) {};

\draw [name path=C3, color=gray!50,radius=3.915cm] (3.06,0.52) circle;

\draw [name path=C4, color=gray!50,radius=3.95cm] (3.13,1.04) circle;

\draw [name path=C5, color=gray!50,radius=4.015cm] (3.23,1.56) circle;

\node at (3,-0.35) (O1) {};
\draw [name path=C2,color=gray!50,radius=3.95cm] (O1) circle;

\path [black, name intersections={of=C1 and C2, by={p1,p2}}] 
      (p1) circle (0.75pt) node [label=-180:$p_2$] {}
      (p2) circle (0.75pt) node [label=-180:$p_1$] {};

\path [black, name intersections={of=horizontal and C2, by={z,t}}] 
      (z) circle (0.75pt) node {};
\path [black, name intersections={of=C3 and C4, by={a,b}}] 
      (a) circle (0.75pt) node [label=-180:$p_3$] {}
      (b) circle (0.75pt) node {};

\path [black, name intersections={of=C4 and C5, by={m,n}}] 
      (m) circle (0.75pt) node [label=-180:$p_4$] {}
      (n) circle (0.75pt) node {};

\draw (p1) -- (s);
\draw (p1) -- (z) -- (m);
\draw (p2) -- (z) -- (a);
\draw [color=darkgray!80,thick] (s) -- (p2) -- (p1) -- (a) -- (m);

\node at (s) [fill, inner sep = 1pt, circle] {};
\node at (p2) [fill, inner sep = 1pt, circle] {};
\node at (p1) [fill, inner sep = 1pt, circle] {};
\node at (a) [fill, inner sep = 1pt, circle] {};
\node at (m) [fill, inner sep = 1pt, circle] {};

\end{tikzpicture}
\begin{tikzpicture}[scale=1.8]
\draw [color=gray!50,radius=0.51cm] (-0.502,-1.05) circle;
\draw [color=gray!50,radius=0.16cm] (-0.86,-1) circle;
\draw [dashed] (-1,-1) -- (0,-1);

\draw [color=darkgray!80,thick,start angle = -90, end angle=180, radius=1cm] (0,-1) arc;
\draw [color=darkgray!80, thick] (-1,0) -- (-1,-0.9) -- (-1.01,-1.05)--(-1.03,-1);

\node at (0, -1) [fill, inner sep = 1pt, circle, label=-45:$t$] (t) {};
\node at (-1.03, -1) [fill, inner sep = 1pt, circle, label=180:$s$] (s) {};

\foreach \angle in {-75, -60, ..., 180}
{
  \node at (\angle:1cm) [fill, inner sep = 1pt, circle] (p) {};
  \draw (p) -- (t);
}
\foreach \d in {-0.15, -0.3, ..., -1.05}
{
  \node at (-1, \d) [fill, inner sep = 1pt, circle] (q) {};
  \draw (q) -- (t); 
}

\node at (-1.01, -1.05) [fill, inner sep = 1pt, circle] (r) {};
\draw (r) -- (q) -- (t) -- (r) -- (s) -- (q);

\end{tikzpicture}
}
\newcommand{\figangles}{
\begin{tikzpicture}[scale=2]
\clip (0,-1.5) rectangle (6,3);
\tkzDefPoint(2.33, 1.24){Op0} %Center of C'_{i-1}
\tkzDefPoint(4.29, 1.09){Op1} %Center of C'_i
\tkzDefPoint(2.33, 0){b0} 
\tkzDefPoint(4.29, 0){b1} 
\tkzDefPoint(3, -0.3){t1} %Center of t_i
\tkzDefPoint(3, -0.3){t0} %Center of t_{i-1}

\tkzDefPoint(1.68, -0.64){q1}
\tkzDefPoint(2.78, -1.1){q2}
\tkzDefPoint(4.5, -0.41){q3}

\defPointsGreenBlue

\tkzDefPointBy[rotation=center Op0 angle 120](p1) \tkzGetPoint{p0}
\tkzDefPointBy[rotation=center Op1 angle -93](p1) \tkzGetPoint{p2}

\tkzDrawPolygon[triangles](p1,p0,q1)
\tkzDrawPolygon[triangles, color=gray](p1,q1,q2)
\tkzDrawPolygon[triangles, color=gray](p1,q2,q3)
\tkzDrawPolygon[triangles](p1,q3,p2)

 \tkzCircumCenter(p1,p0,q1)\tkzGetPoint{O0}
 \tkzCircumCenter(p1,q1,q2)\tkzGetPoint{O1}
 \tkzCircumCenter(p1,q2,q3)\tkzGetPoint{O2}
 \tkzCircumCenter(p1,q3,p2)\tkzGetPoint{O3}

\tkzDefLine[parallel=through O0](s,t)\tkzGetPoint{O0h}
\tkzDefLine[parallel=through O1](s,t)\tkzGetPoint{O1h}
\tkzDefLine[parallel=through O2](s,t)\tkzGetPoint{O2h}
\tkzDefLine[parallel=through O3](s,t)\tkzGetPoint{O3h}

\tkzInterLC(O0,O0h)(O0,p1) \tkzGetFirstPoint{w0}
\tkzInterLC(O1,O1h)(O1,p1) \tkzGetFirstPoint{w01}
\tkzInterLC(O2,O2h)(O2,p1) \tkzGetFirstPoint{w02}

\tkzInterLC(O3,O3h)(O3,p1) \tkzGetFirstPoint{w1}

\tkzDrawCircle[emptyCircle](O0,p1)\tkzDrawPoint(O0) \tkzLabelPoint[below left](O0){$O_{i-1}$}
\tkzDrawCircle[emptyCircle](O1,p1)\tkzDrawPoint(O1) \tkzLabelPoint[below](O1){$O^1_i$}
\tkzDrawCircle[emptyCircle](O1,p1)\tkzDrawPoint(O2) \tkzLabelPoint[below right](O2){$O^2_i$}

\tkzDrawCircle[emptyCircle](O0,p1)\tkzDrawPoint(O3) \tkzLabelPoint[right](O3){$O_{i}$}
\tkzDrawCircle[emptyCircle](O1,p1)
\tkzDrawCircle[emptyCircle](O2,p1)
\tkzDrawCircle[emptyCircle](O3,p1)

\tkzDrawSegment[dashed](Op0,O0)
\tkzDrawSegment[dashed](O0,O1)
\tkzDrawSegment[dashed](O1,O2)
\tkzDrawSegment[dashed](O2,O3)

\tkzDrawSegment[dotted](pp0,Op0)\tkzDrawSegment[dotted](Op0,p1)\tkzMarkAngle[size=0.15](p1,Op0,pp0)
\tkzDrawSegment[dotted](w0,O0)\tkzDrawSegment[dotted](O0,p1)\tkzMarkAngle[size=0.15](p1,O0,w0)
\tkzDrawSegment[dotted](w01,O1)\tkzDrawSegment[dotted](O1,p1)\tkzMarkAngle[size=0.15](p1,O1,w01)
\tkzDrawSegment[dotted](w02,O2)\tkzDrawSegment[dotted](O2,p1)\tkzMarkAngle[size=0.15](p1,O2,w02)
\tkzDrawSegment[dotted](w1,O3)\tkzDrawSegment[dotted](O3,p1)\tkzMarkAngle[size=0.15](p1,O3,w1)
\tkzDrawSegment[dotted](pp1,Op1)\tkzDrawSegment[dotted](Op1,p1)\tkzMarkAngle[size=0.15](p1,Op1,pp1)

\tkzDefPointBy[homothety=center O0 ratio 10 ](Op0) \tkzGetPoint{Opp0}
\tkzDrawSegment[dashed](O0,Opp0)
\tkzDefPointBy[homothety=center O3 ratio 10 ](Op1) \tkzGetPoint{Opp1}
\tkzDrawSegment[dashed](O3,Opp1)

\tkzDrawCircle[Cprime](Op0,b0)

\tkzDrawCircle[Cprime](Op0,b0)
\tkzLabelCircle[Cprime,above=8pt](Op0,b0)(210){$\Cpi{i-1}$}
\tkzDrawCircle[Cprime](Op1,b1)
\tkzLabelCircle[Op1,above right=4pt](Op1,b1)(60){$\Cpi{i}$}

\tkzDrawPoint(pp0) \tkzLabelPoint[ left](pp0){$w'_{i-1}$}

\tkzDrawPoint(w0) \tkzLabelPoint[left](w0){$w_{i-1}$}
\tkzDrawPoint(w01) \tkzLabelPoint[left](w01){$w^1_{i}$}
\tkzDrawPoint(w02) \tkzLabelPoint[left](w02){$w^2_{i}$}
\tkzDrawPoint(w1) \tkzLabelPoint[left](w1){$w_{i}$}
\tkzDrawPoint(Op0) \tkzLabelPoint[below left](Op0){${\color{\red} O'_{i-1}}$}

\tkzDrawPoint(pp1) \tkzLabelPoint[left](pp1){$w'_{i}$}

\tkzDrawPoint(Op1) \tkzLabelPoint[above right](Op1){${\color{\red} O'_{i}}$}
\tkzDrawPoint(p0) \tkzLabelPoint[above left](p0){$p_{i-1}$}
\tkzDrawPoint(p1) \tkzLabelPoint[above ](p1){$p_i$}
\tkzDrawPoint(p2) \tkzLabelPoint[above right](p2){$p_{i+1}$}

\tkzDrawPoint(q1) \tkzLabelPoint[below](q1){$q_{1}$}
\tkzDrawPoint(q2) \tkzLabelPoint[below](q2){$q_{2}$}
\tkzDrawPoint(q3) \tkzLabelPoint[below](q3){$q_{3}$}

\end{tikzpicture}
}

\newcommand{\figexample}{
\begin{tikzpicture}[scale=3]
\clipExample
\foreach \x\y/\name/\etiq/\pos  in \nodeList
  {\tkzDefPoint(\x/100,\y/100){\name}}

\tkzComputeAllPoints

\tkzDrawLine(s,t)

\foreach \A/\B/\C/\cw/\cpos/\cposangle/\cnum/\ctype  in \routingPath
  {    
   \tkzDrawCircle[emptyCircle](G\A\B\C,\A)
   \ifthenelse{\equal{\cw}{1}}{\tkzDrawArc[path](G\A\B\C,\B)(\A)}{\tkzDrawArc[path](G\A\B\C,\A)(\B)}
       \tkzDrawCircle[color=\red](Gp\A\B\C,\A)
       \tkzLabelCircle[\cpos=4pt,color=\red](Gp\A\B\C,\A)(\cposangle){$\Cpi{\cnum}$}
       \ifthenelse{\equal{\cw}{1}}{\tkzDrawArc[dPCurve](Gp\A\B\C,\B)(primep\cnum)}{\tkzDrawArc[dPCurve](Gp\A\B\C,primep\cnum)(\B)}}

\tkzDrawSegment[dPpCurve](s0,t0)
\tkzDrawSegment[dPpCurve](s1,t1)
\tkzDrawSegment[dPpCurve](s2,t2)
\tkzDrawSegment[dPpCurve](s3,t4)

    \foreach \i in {0,...,4} {
      \tkzDrawSegment[dPCurve](primep\i,s\i)
   }
      \tkzDrawPoint[size=2](p1) \tkzLabelPoint[above right](p1){$p_1$}
      \tkzDrawPoint[size=2](p2) \tkzLabelPoint[above left](p2){$p_2$}
      \tkzDrawPoint[size=2](p3) \tkzLabelPoint[above=3pt](p3){$p_3$}
      \tkzDrawPoint[size=2](p4) \tkzLabelPoint[above left=3pt](p4){$p_4$}

      \tkzDrawPoint[size=2](primep3)
      \tkzLabelPoint[left](primep3){$w'_3$}
\tkzDrawPoint[size=2](sp1) \tkzLabelPoint[above left](sp1){$s_1$}
\tkzDrawPoint[size=2](sp2) \tkzLabelPoint[above left](sp2){$s_2$}
      \tkzDrawPoint[size=2](s)
      \tkzLabelPoint[left](s){$s$}
      \tkzDrawPoint[size=2](t)
      \tkzLabelPoint[below right](t){$t$}
\end{tikzpicture}
}

\newcommand{\figcaseAOneB}{
\begin{minipage}[t]{0.48\linewidth}
\begin{tikzpicture}[scale=1.2]
\clip (-1,-0.5) rectangle (3.5,3.5);
\tkzDefPoint(1.5, 1.2){Op0} %Center of C'_{i-1}
\tkzDefPoint(2.4, 2){Op1} %Center of C'_i
\tkzDefPoint(1.5, -0.3){b0} %Center of C'_{i-1}
\tkzDefPoint(2.4, 0){b1} %Center of C'_i
\tkzDefPoint(3, -0.3){t1} %Center of t_i
\tkzDefPoint(3, -0.3){t0} %Center of t_{i-1}

\defPointsGreenBlue

\tkzDrawCircle[Cprime](Op0,b0)
\tkzLabelCircle[Cprime,above=8pt](Op0,b0)(150){$\Cpi{i-1}$}

\LHS
\RHS

\tkzDrawPoint[size=2](pp0) \tkzLabelPoint[above left](pp0){$w'_{i-1}$}
\tkzDrawPoint[size=2](sp0) \tkzLabelPoint[below left](sp0){$s_{i-1}$}
\tkzDrawPoint[size=2](p1) \tkzLabelPoint[below right](p1){$p_i$}
\tkzDrawPoint[size=2](sp1) \tkzLabelPoint[below right](sp1){$s_{i}$}
\end{tikzpicture}
\end{minipage}
\begin{minipage}[t]{0.48\linewidth}
\begin{tikzpicture}[scale=1.2]
\clip (-1,-0.5) rectangle (3.5,3.5);
\tkzDefPoint(1.5, 1.2){Op0} %Center of C'_{i-1}
\tkzDefPoint(2.4, 2){Op1} %Center of C'_i
\tkzDefPoint(1.5, -0.3){b0} %Center of C'_{i-1}
\tkzDefPoint(2.4, 0){b1} %Center of C'_i
\tkzDefPoint(3, -0.3){t1} %Center of t_i
\tkzDefPoint(3, -0.3){t0} %Center of t_{i-1}

\defPointsGreenBlue

\tkzDrawCircle[Cprime](Op0,b0)
\tkzLabelCircle[Cprime,above=8pt](Op0,b0)(150){$\Cpi{i-1}$}

\LHS

\tkzDrawSegment[RHS](X,pps2)
\tkzDrawSegment[RHS](sp0,X)
\tkzDrawArc[RHS](Os2,p1)(pps2)

\tkzDrawPoint[size=2](pp0) \tkzLabelPoint[above left](pp0){$w'_{i-1}$}
\tkzDrawPoint[size=2](sp0) \tkzLabelPoint[below left](sp0){$s_{i-1}$}
\tkzDrawPoint[size=2](p1) \tkzLabelPoint[below right](p1){$p_i$}
\tkzDrawPoint[size=2](sp1) \tkzLabelPoint[below right](sp1){$s_{i}$}
\tkzDrawPoint[size=2](X) \tkzLabelPoint[left](X){$X$}
\end{tikzpicture}
\end{minipage}
}

\newcommand{\figcaseAAOne}{
\begin{tikzpicture}[scale=1.4]
\clip (-0.5,-1.5) rectangle (7.5,6.5);
\tkzDefPoint(2.33, 1.18){Op0} %Center of C'_{i-1}
\tkzDefPoint(3.6, 1.57){Op1} %Center of C'_i
\tkzDefPoint(2.33, -0.85){b0} %Center of C'_{i-1}
\tkzDefPoint(3.6, 0){b1} %Center of C'_i
\tkzDefPoint(3, -0.3){t1} %Center of t_i
\tkzDefPoint(3, -0.3){t0} %Center of t_{i-1}

\defPointsGreenBlue

\tkzDefPointBy[translation= from b1 to bp1](Op1) \tkzGetPoint{Op1bis}
\tkzDefPointBy[translation= from b1 to bp1](p1) \tkzGetPoint{p1bis}

\tkzDrawCircle[Cprime](Op0,b0)
\tkzLabelCircle[Cprime,above=8pt](Op0,b0)(210){${\color{\blue}\pathf{B}}$}
\tkzDrawCircle[Cprime](Op1,b1)
\tkzLabelCircle[Cprime,right=4pt](Op1bis,bp1)(90){${\color{\red}\pathf{R}}$}
\tkzLabelCircle[Cprime,right=4pt](Op1,b1)(90){$\Cpi{i}$}
\tkzLabelCircle[Cprime,below left=4pt](Op1bis,bp1)(0){$C''_i$}
\tkzLabelCircle[Op0Y,left=4pt](Op0Y,Z)(120){$\pathf{Y}$}
\tkzLabelCircle[Os3,above right=4pt](Os3,bp1)(90){$\pathf{G}$}
\tkzLabelCircle[Op1,left=4pt](Op1,b1)(240){$\arcP{i}{s_{i},p_i}$}
\tkzDrawArc[LHS](Op0,p1)(pp0)
\tkzDrawSegment[LHS](pp0,spp0)

\tkzDrawArc[RHS](Op1,p1)(pp1)
\tkzDrawSegment[RHS](pp1,sp1)
\tkzDrawArc[RHS](Os3,bp1)(pps3)
\tkzDrawSegment[RHS](pps3,spp0)
\tkzDrawSegment[dashed](spp0,Z)

\tkzDrawArc[RHS,color=\red](Op1bis,bp1)(p1bis)
\tkzDrawCircle[Cprime](Op1bis,bp1)

\tkzDrawSegment[RHS,color=\yellow](p1,pp0Y)
\tkzDrawArc[RHS,color=\yellow](Op0Y,Z)(pp0Y)

\tkzDrawPoint[size=2](Z) \tkzLabelPoint[above right](Z){$Z$}
\tkzDrawPoint[size=2](p1bis) \tkzLabelPoint[below](p1bis){$W$}
\tkzDrawPoint[size=2](pp0) \tkzLabelPoint[above left](pp0){$w'_{i-1}$}
\tkzDrawPoint[size=2](b0) \tkzLabelPoint[below right](b0){$b_{i-1}$}
\tkzDrawPoint[size=2](spp0) \tkzLabelPoint[below left](spp0){$X$}
\tkzDrawPoint[size=2](sp0) \tkzLabelPoint[above left](sp0){$s_{i-1}$}
\tkzDrawPoint[size=2](pp1) \tkzLabelPoint[above left](pp1){$w'_{i}$}
\tkzDrawPoint[size=2](b1) \tkzLabelPoint[above right](b1){$b_{i}$}
\tkzDrawPoint[size=2](bp1) \tkzLabelPoint[below right](bp1){$Y$}
\tkzDrawPoint[size=2](p1) \tkzLabelPoint[above right=4pt](p1){$p_i=t_{i-1}$}
\tkzDrawPoint[size=2](sp1) \tkzLabelPoint[below right](sp1){$s_{i}$}
\end{tikzpicture}
}

\newcommand{\figcaseBallTwo}{
\begin{minipage}[t]{0.4\linewidth}
\begin{tikzpicture}[scale=2.2]
\clip (-0.5,-2) rectangle (2.5,1.5);
\tkzDefPoint(1, -0.5){Op0} %Center of C'_{i-1}
\tkzDefPoint(0.75, 0.24){Op1} %Center of C'_i
\tkzDefPoint(1, -1.5){b0} %Center of C'_{i-1}
\tkzDefPoint(0.75, 0){b1} %Center of C'_i
\tkzDefPoint(0.9, -1.6){t1} %Center of t_i

\defPointsGreenBlue

\tkzDrawCircle[Cprime](Op0,b0)
\tkzLabelCircle[Cprime,above=8pt](Op0,b0)(150){$\Cpi{i-1}$}

\tkzDrawArc[LHS](Op0,p1)(pp0)
\tkzDrawSegment[LHS](pp0,tpp0)
\tkzDrawSegment[LHS,style=dashdotted](tpp0,p1)

\RHS

\tkzDrawSegment[RHS,style=dashdotted](tpp1,t1)
\tkzDrawSegment[RHS](tp1,tpp1)

\tkzDrawPoint[size=1](tp1) \tkzLabelPoint[below right](tp1){$t'_{i-1}$}
\tkzDrawPoint[size=1](t1) \tkzLabelPoint[below right](t1){$t_{i}$}
\tkzDrawPoint[size=1](pp0) \tkzLabelPoint[below left](pp0){$w'_{i-1}$}
\tkzDrawPoint[size=1](sp0) \tkzLabelPoint[below left](sp0){$s_{i-1}$}
\tkzDrawPoint[size=1](p1) \tkzLabelPoint[below right](p1){$p_i$}
\tkzDrawPoint[size=1](sp1) \tkzLabelPoint[below left](sp1){$s_{i}$}
\end{tikzpicture}
\end{minipage}
\begin{minipage}[t]{0.4\linewidth}
\begin{tikzpicture}[scale=2.2]
\clip (-0.5,-2) rectangle (2.5,1.5);

\tkzDefPoint(1, -0.5){Op0} %Center of C'_{i-1}
\tkzDefPoint(0.75, 0.24){Op1} %Center of C'_i
\tkzDefPoint(1, -1.5){b0} %Center of C'_{i-1}
\tkzDefPoint(0.75, 0){b1} %Center of C'_i
\tkzDefPoint(0.9, -1.6){t1} %Center of t_i

\defPointsGreenBlue

\tkzDrawCircle[Cprime](Op0,b0)
\tkzLabelCircle[Cprime,above=8pt](Op0,b0)(150){$\Cpi{i-1}$}

\tkzDrawArc[LHS](Op0,p1)(pp0)
\tkzDrawSegment[LHS](pp0,tpp0)
\tkzDrawSegment[LHS,style=dashdotted](tpp0,p1)

\tkzDrawSegment[RHS](X,pps2)
\tkzDrawArc[RHS](Op12,X)(pp12)
\tkzDrawSegment[RHS](pp12,sp0)
\tkzDrawArc[RHS](Os2,p1)(pps2)
\tkzDrawSegment[RHS,style=dashdotted](tpp1,t1)
\tkzDrawSegment[RHS](tp1,tpp1)

\tkzDrawPoint[size=1](tp1) \tkzLabelPoint[below right](tp1){$t'_{i-1}$}
\tkzDrawPoint[size=1](t1) \tkzLabelPoint[below right](t1){$t_{i}$}
\tkzDrawPoint[size=1](pp0) \tkzLabelPoint[below left](pp0){$w'_{i-1}$}
\tkzDrawPoint[size=1](sp0) \tkzLabelPoint[below left](sp0){$s_{i-1}$}
\tkzDrawPoint[size=1](p1) \tkzLabelPoint[below right](p1){$p_i=t_{i-1}$}
\tkzDrawPoint[size=1](sp1) \tkzLabelPoint[below left](sp1){$s_{i}$}
\tkzDrawPoint[size=1](X) \tkzLabelPoint[left](X){$X$}
\end{tikzpicture}
\end{minipage}
}

\newcommand{\figcaseBall}{
\begin{tikzpicture}[scale=3]
\clip (-0.5,-1.6) rectangle (2.5,1.2);
\tkzDefPoint(1, -0.5){Op0} %Center of C'_{i-1}
\tkzDefPoint(0.75, 0.24){Op1} %Center of C'_i
\tkzDefPoint(1, -1.5){b0} %Center of C'_{i-1}
\tkzDefPoint(0.75, 0){b1} %Center of C'_i
\tkzDefPoint(3, 0){t1} %Center of t_i

\tkzInterCC(Op0,b0)(Op1,b1) \tkzGetFirstPoint{p1}
\tkzDefPoint(-0.2, 0){s} 
\tkzDefPoint(2, 0){t} 
\tkzDefMidPoint(p1,p1)\tkzGetPoint{t0}
\tkzInterLC(s,t)(Op0,b0) \tkzGetPoints{l0}{t1}

\defPointsGreenBlue
\tkzDefLine[parallel=through p1](s,t) \tkzGetPoint{horip1}
\tkzInterLC(p1,horip1)(Op0,b0)\tkzGetPoints{toto}{rp1}
\tkzDefPointBy[projection=onto s--t](rp1) \tkzGetPoint{rp1bis}
\tkzDrawCircle[Cprime](Op0,b0)
\tkzLabelCircle[Cprime,above=8pt](Op0,b0)(30){$\Cpi{i-1}$}
\tkzDrawCircle[Cprime](Op1,b1)
\tkzLabelCircle[Cprime,right=4pt](Op1,b1)(90){$\Cpi{i}$}

\tkzDefPointBy[rotation=center Op0 angle 180](pp0) \tkzGetPoint{ppp0}

\tkzDrawSegment[dashed](Op0,p1)
\tkzDrawSegment[dashed](Op0,pp0)

\tkzMarkAngle[size=0.2](p1,Op0,pp0) \tkzLabelAngle[pos=0.4](p1,Op0,pp0){$\theta$}
\tkzMarkAngle[size=0.15](p1,Op1,pp1) \tkzLabelAngle[pos=0.19](p1,Op1,pp1){$\alpha$}
\tkzDrawSegment[dashed](Op1,p1)
\tkzDrawSegment[dashed](Op1,pp1)
\tkzDrawSegment[dashed](rp1,rp1bis)

\LHS
\RHS

\tkzDrawSegment[LHS,style=dashdotted](s0,t0)
\tkzDrawSegment[LHS](s0,sp0)
\tkzDrawSegment[RHS,style=dashdotted](tp1,rp1bis)

\tkzDrawPoint[size=1](pp0) \tkzLabelPoint[above left](pp0){$w'_{i-1}$}
\tkzDrawPoint[size=1](tp1) \tkzLabelPoint[below right](tp1){$t'_{i-1}$}
\tkzDrawPoint[size=1](sp0) \tkzLabelPoint[above left](sp0){$s_{i-1}$}
\tkzDrawPoint[size=1](Op0) \tkzLabelPoint[above right](Op0){$O'_{i-1}$}
\tkzDrawPoint[size=1](pp1) \tkzLabelPoint[below left](pp1){$w'_{i}$}
\tkzDrawPoint[size=1](p1) \tkzLabelPoint[above left](p1){$p_i=t_{i-1}$}
\tkzDrawPoint[size=1](sp1) \tkzLabelPoint[below](sp1){$s_{i}$}
\tkzDrawPoint[size=1](Op1) \tkzLabelPoint[above](Op1){$O'_{i}$}

\tkzDrawPoint[size=1](rp1) \tkzLabelPoint[above right](rp1){$p'_{i}$}
\end{tikzpicture}
}

\newcommand{\figLBLTwo}{
\includegraphics[trim=0cm 7.5cm 0cm 6.5cm,clip, width=0.7\textwidth]{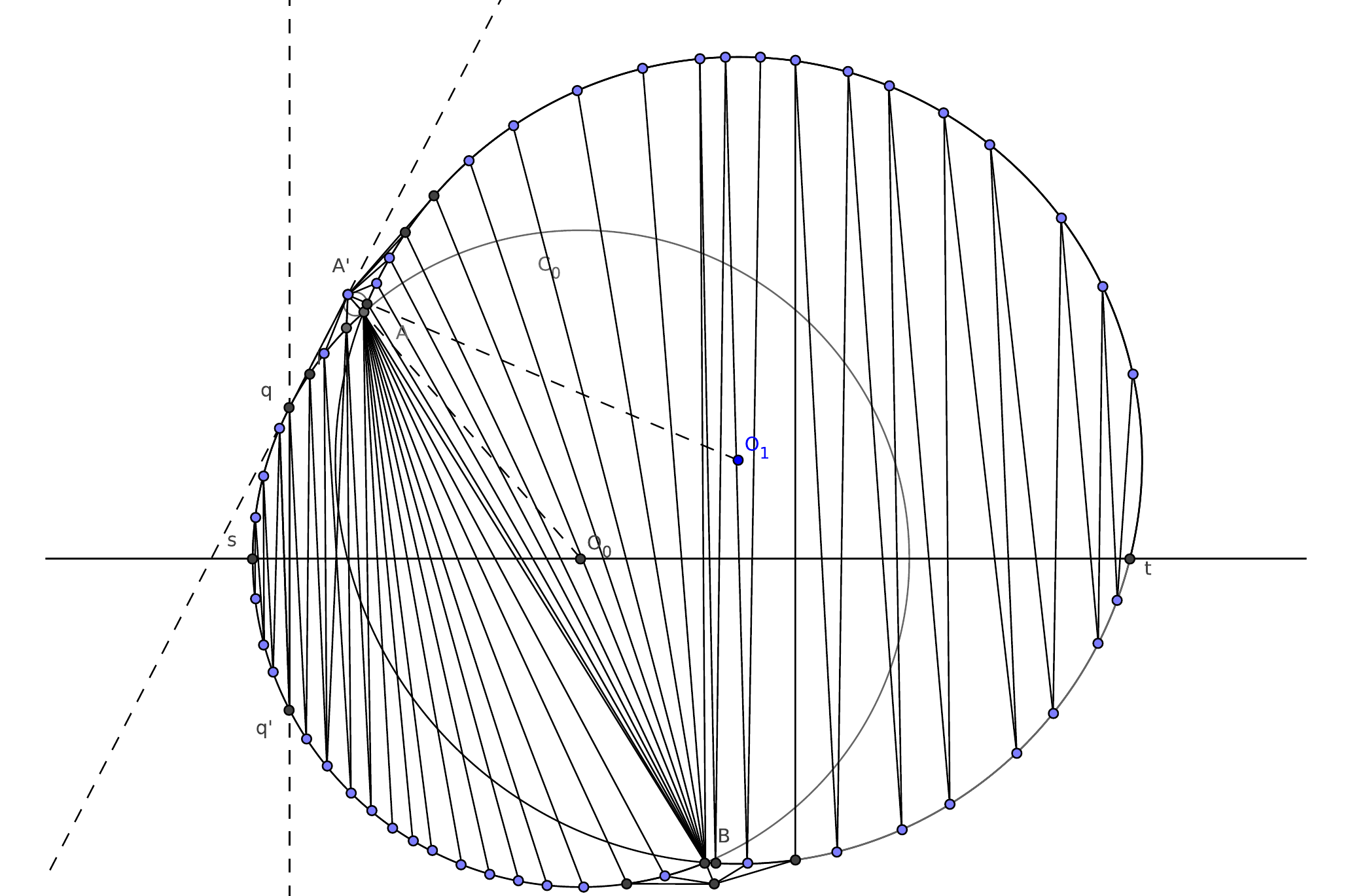}
}

\newcommand{\figLBLOne}{
\includegraphics[trim=0cm 6cm 0cm 8.2cm,clip, width=0.5\textwidth]{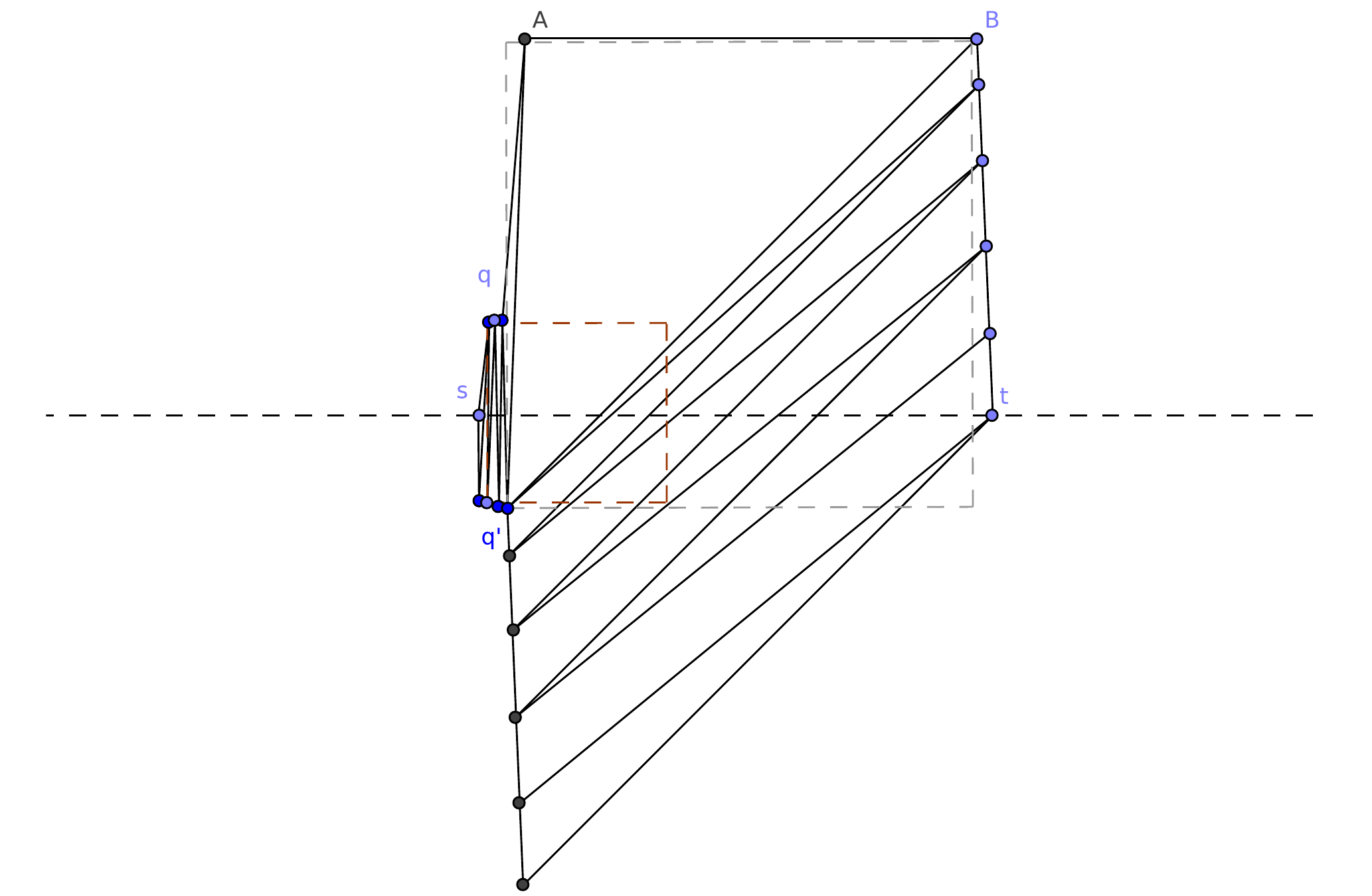}
}

\title{Upper and Lower Bounds for Competitive Online Routing on Delaunay Triangulations\footnote{N. Bonichon partially supported by ANR grant JCJC EGOS ANR-12-JS02-002-01 and LIRCO. P. Bose partially supported by NSERC. A. van Renssen supported by JST, ERATO, Kawarabayashi Large Graph Project.}}

\author[1]{Nicolas Bonichon}
\author[2]{Prosenjit Bose}
\author[2]{Jean-Lou De Carufel}
\author[3]{Ljubomir Perkovi\'{c}}
\author[4]{Andr\'e van Renssen}
\affil[1]{LaBRI, Bordeaux University, 
  Bordeaux, France\\
  \texttt{bonichon@labri.fr}}

\affil[2]{School of Computer Science, Carleton University, 
 Ottawa, Canada\\
   \texttt{jit@scs.carleton.ca, jdecaruf@cg.scs.carleton.ca}}

\affil[3]{School of Computing, DePaul University,
  Chicago, USA\\
  \texttt{lperkovic@cs.depaul.edu}}

\affil[4]{National Institute of Informatics,
 Tokyo, Japan \\
  \texttt{andre@nii.ac.jp}}

\authorrunning{N. Bonichon, P. Bose, J.-L. De Carufel, L. Perkovi\'{c}, A. van Renssen}

\Copyright{N. Bonichon, P. Bose, J.-L. De Carufel, L. Perkovi\'{c}, A. van Renssen}

\subjclass{I.3.5 Computational Geometry and Object Modeling, G.2.2 Graph Theory}

\keywords{Delaunay triangulation, online routing}

\begin{document}

\maketitle

\begin{abstract}
Consider a weighted graph $G$ where vertices are points in the plane and edges are line segments. The weight of each edge is the Euclidean distance between its two endpoints. A routing algorithm on $G$ has a \emph{competitive ratio} of $c$ if the length of the path produced by the algorithm from any vertex $s$ to any vertex $t$ is at most $c$ times the length of the shortest path from $s$ to $t$ in $G$. If the length of the path is at most $c$ times the Euclidean distance from $s$ to $t$, we say that the routing algorithm on $G$ has a \emph{routing ratio} of $c$.

We present an online routing algorithm on the Delaunay triangulation with competitive and routing ratios of $5.90$. This improves upon the best known algorithm that has competitive and routing ratio $15.48$. The algorithm is a generalization of the deterministic $1$-local routing algorithm by Chew on the $L_1$-Delaunay triangulation. When a message follows the routing path produced by our algorithm, its header need only contain the coordinates of $s$ and $t$. This is an improvement over the currently known competitive routing algorithms on the Delaunay triangulation, for which the header of a message must additionally contain partial sums of distances along the routing path.

We also show that the routing ratio of any deterministic $k$-local algorithm is at least $1.70$ for the Delaunay triangulation and $2.70$ for the $L_1$-Delaunay triangulation. In the case of the $L_1$-Delaunay triangulation, this implies that even though there exists a path between two points $x$ and $y$ whose length is at most $2.61|[xy]|$
(where $|[xy]|$ denotes the length of the line segment $[xy]$), it is not
always possible to route a message along a path of length less than $2.70|[xy]|$. From these bounds on the routing ratio, we derive lower bounds on the competitive ratio of $1.23$ for Delaunay triangulations and $1.12$ for $L_1$-Delaunay triangulations.
\end{abstract}

%\section{Introduction}

\section{Introduction}

Navigation in networks or in graphs is fundamental in computer science.
It leads to applications in a number of fields, namely geographic information
systems, urban planning, robotics, and communication networks to name only a
few. Navigation often occurs in a geometric setting that can be modeled using
a {\em geometric graph} which is defined as a weighted graph $G$ whose
vertices are points in the plane and edges are line segments.
Let the weight of each edge be the Euclidean distance between its two endpoints.
Navigation is then simply the problem of finding a path in $G$ from a source
vertex $s$ to a target vertex $t$.
When complete information about the graph is available,
numerous path finding algorithms exist for weighted graphs
(e.g., Dijkstra's algorithm~\cite{dijkstra1959note}).

The problem is more challenging when only \emph{local} information is provided.
To illustrate this, suppose that a message is traveling along edges of $G$.
We are interested in a routing algorithm that makes forwarding decisions based on limited information related to the current position of the message in $G$.
If the message is currently at vertex $v$, such information
could be limited to the coordinates of $v$ and its neighbors in the graph.
Such an algorithm is called {\em local} or {1-local}. More
generally,  when the coordinates of neighbours that are at most $k$ hops away
from $v$ are available, we say that the algorithm is \emph{$k$-local}.
A routing algorithm that uses only local information is called an \emph{online} routing algorithm.
%Different hypotheses can be made about the information that is stored in the header of the message.
%Typically,
%we suppose that the header contains the coordinates of $s$ and $t$.
%However,
%is some cases,
%we need to store additional information for the algorithm to succeed.
Given a constant $c \geq 1$, an online routing algorithm has a \emph{competitive ratio} of $c$ or is $c$-competitive if the length of the path produced by the algorithm from any vertex $s$ to any vertex $t$ is at most $c$ times the length of the shortest path from $s$ to $t$ in $G$.
If the length of the path is at most $c|[st]|$,
where $|[st]|$ is the Euclidean length of the line segment $[st]$,
we say that the routing algorithm has a \emph{routing ratio} of $c$. 
Since $|[st]|$ is a lower bound on the length of the shortest path from $s$ to $t$ in $G$,
the routing ratio is an upper bound on the competitive ratio.

Competitive online routing is not always possible
(refer to~\cite{DBLP:journals/ijcga/BoseBCDFLMM02,DBLP:journals/winet/DurocherKN10} for instance) and even when it is, a small competitive ratio may not
be. In this paper, we are interested in competitive online routing
algorithms for  classes of geometric graphs that have ``good paths''.
A graph $G$ is a \emph{$\kappa$-spanner}
(or has a \emph{spanning ratio} of $\kappa$)
when for any pair of vertices $u$ and $v$ in $G$,
there exists a path in $G$ from $u$ to $v$
with length at most $\kappa|[uv]|$
(see~\cite{DBLP:conf/wads/BarbaBCRV13,
DBLP:conf/compgeom/BarbaBDFKORTVX14,
DBLP:journals/jocg/BoseCCS10,
DBLP:conf/cccg/BoseCMRV12,
DBLP:journals/ijcga/BoseDDOSSW12,
DBLP:journals/comgeo/BoseMRV15,
DBLP:conf/wads/BoseRV13,
Xia13}
for instance).
In several cases,
the proof of existence of these paths rely on full knowledge of the graph.
Therefore,
a natural question is to ask whether we can construct these paths
using only local information.
We do not have a general answer to that question.
Nevertheless,
there exist $c$-competitive online routing algorithms for several of these classes of geometric graphs
(see~\cite{
DBLP:journals/comgeo/BoseCCSX11,
BDDT14,
BFRV2012Routing,
DBLP:conf/wads/BoseRV13,
Che86,
Che89}).
In this paper, we focus on the most important geometric graph, the Delaunay
triangulation.

A Delaunay triangulation is a geometric graph $G$ such that there is an edge
between two vertices $u$ and $v$ if and only if there exists a {\em circle}
with $u$ and $v$ on its boundary that contains no other vertex of $G$.
Dobkin et al.~\cite{DFS90} were the first to prove that the Delaunay 
triangulation is a spanner.
Xia~\cite{Xia13} proved that the spanning ratio of the Delaunay triangulation 
is at most $1.998$ which currently best known the best upper bound on the
spanning ratio. The best lower bound, by Xia et al.~\cite{xia2011toward}, is
$1.593$. If {\em circle} is replaced with {\em equilateral triangle}
in the definition of the Delaunay triangulation, then a different
triangulation is defined: the $TD$-Delaunay triangulation. Chew proved that
the $TD$-Delaunay triangulation~\cite{Che89} is a $2$-spanner
and that the constant $2$ is tight. If we replace {\em circle} with
{\em square} then yet another triangulation is defined: the $L_1$- or the
$L_{\infty}$-Delaunay triangulation, depending on the orientation of the square.
Bonichon et al.~\cite{BGHP12} proved that the $L_1$- and the 
$L_{\infty}$-Delaunay triangulations are $\sqrt{4+2\sqrt{2}}$-spanners and 
that the constant is also tight. 

Ideally, we would like the routing ratio to be identical to the spanning ratio.
In the case of $TD$-Delaunay triangulations, Bose et al.~\cite{BFRV2012Routing}
found an online routing algorithm that has competitive and routing ratios of
$\frac{5}{\sqrt{3}}$. They also showed lower bounds of
$\frac{5}{\sqrt{3}}$ on the routing ratio and of $\frac{5}{3}$ on the
competitive ratio. In his seminal paper, Chew~\cite{Che89} described an
online routing algorithm on the $L_1$-Delaunay triangulation that has
competitive and routing ratios of $\sqrt{10}\approx 3.162$.
We show in this paper that there is separation between the spanning ratio and
the routing ratio in the case of the $L_1$ and $L_{\infty}$-Delaunay
triangulations. We show lower bounds of $2.707$ on the routing ratio
and of $1.122$ on the competitive ratio (Theorem~\ref{thm:L1Lower}). 
In this paper, we also present an online routing algorithm on the
Delaunay triangulation that has competitive and routing ratios of $5.90$ 
(Theorem~\ref{thm:main}).
This improves upon the previous best known algorithm that has competitive and
routing ratios of $15.48$~\cite{BDDT14}.
Our algorithm is a generalization of the deterministic $1$-local routing 
algorithm by Chew on the $L_1$-Delaunay triangulation~\cite{Che86} and the 
$TD$-Delaunay triangulation~\cite{Che89}.
Although the generalization of Chew's routing algorithm to Delaunay
triangulation is natural, the analysis of its routing ratio is non-trivial
and relies on new techniques. An advantage of Chew's routing algorithm
is that it does not require the message header to contain any information other
than the coordinates of $s$ and $t$. All previously known competitive routing
algorithms on the Delaunay triangulation \cite{BDDT14,DBLP:journals/siamcomp/BoseM04} require header to store partial sums of distances along the routing
path. in the header of the message. See Table~\ref{ta:related} for a summary
of these results.

\begin{table}[!thb]
\begin{center}
\begin{tabular}{|l|c|c|c|}
\hline {\bf Shape} &  {\bf triangle} & {\bf square} & {\bf circle} \\ \hline \hline
spanning ratio UB &  2~\cite{Che89} & 2.61~\cite{BGHP12} & $1.998$~\cite{Xia13}  \\ \hline
spanning ratio LB &  2~\cite{Che89} & 2.61~\cite{BGHP12} & $1.593$~\cite{xia2011toward}  \\ \hline
routing ratio UB &$5/\sqrt{3} \approx 2.89$~\cite{BFRV2012Routing}  & $\sqrt{10} \approx 3.16$~\cite{Che86} & $1.185 + 3\pi/2 \approx 5.90$ (Thm~\ref{thm:main})\\ \hline
routing ratio LB &  $5/\sqrt{3} \approx 2.89$~\cite{BFRV2012Routing} & 2.707 (Thm~\ref{thm:L1Lower}) & 1.701 (Thm~\ref{thm:L2Lower})\\ \hline
competitiveness LB & $5/3 \approx 1.66$~\cite{BFRV2012Routing}  & 1.1213 (Thm~\ref{thm:L1Lower}) & 1.2327 (Thm~\ref{thm:L2Lower})\\ \hline
\end{tabular}
\end{center}
\caption{Upper and lower bounds on the spanning ratio and the routing ratio on Delaunay triangulations defined by different empty shapes. We also provide lower bounds on the competitiveness of $k$-local deterministic routing algorithms on Delaunay triangulations.
\label{ta:related}}
\end{table}
%Routing in the Delaunay triangulation has been studied from different perspectives~\cite{BDDT14,
%BFRV2012Routing,
%DBLP:journals/siamcomp/BoseM04,
%DBLP:journals/corr/BroutinDH14,
%Che86,Che89,
%DBLP:journals/tpds/SiZ12}.

\section{Chew's Routing Algorithm}
\label{sec:chew}
In this section we present the routing algorithm. This algorithm is a
natural adaptation to Delaunay triangulations of Chew's routing algorithm
originally designed for $L_1$-Delaunay
triangulations~\cite{Che86} and subsequently adapted for $TD$-Delaunay
triangulations~\cite{Che89}.

We consider the Delaunay triangulation defined on a finite set of points $P$
in the plane.
In this paper,
we denote the source of the routing path by $s\in P$
and its destination by $t\in P$.
We assume that an orthogonal coordinate system consisting of
a horizontal $x$-axis and a vertical $y$-axis exists and we denote by $x(p)$
and $y(p)$ the $x$- and $y$-coordinates of any point $p$ in the plane. 
We denote the line supported by two points $p$ and $q$ by $pq$,
and the line segment with endpoints $p$ and $q$ by $[pq]$.
Without
loss of generality, we assume that $y(s) = y(t) = 0$ and $x(s) < x(t)$.

When routing from $s$ to $t$, we consider only 
(the vertices and edges of) the triangles of the Delaunay
triangulation that intersect $[st]$.
Without loss of generality,
if a vertex (other than $s$ and $t$) is on $[st]$,
we consider it to be slightly above $st$.
Therefore,
the triangles that intersect $[st]$ can be ordered from left to right.
Notice that all vertices (other than $s$ and $t$)
from this ordered set of triangles belong to at least $2$ of these triangles.

The routing algorithm can be described as follows. When we
reach a vertex $p_i$ (initially $p_0 = s$), we consider the
{\em rightmost} triangle $T_i$ that has $p_i$ as a vertex.
Let $x$ and $y$ be the other two vertices of $T_i$ and denote by $C_i$ the
the circle circumscribing $T_i$. 
Let $w_i$ ($w$ as in \emph{west}) be the
leftmost point of $C_i$ and let $r_i$ be the rightmost intersection of
$C_i$ with $[st]$. The line segment $[w_ir_i]$ splits $C_i$ in two arcs: the
{\em upper} one, defined by the clockwise walk along $C_i$ from $w_i$ to $r_i$
and the {\em lower} one, defined by the counterclockwise walk along $C_i$ from
$w_i$ to $r_i$. Both arcs include points $w_i$ and $r_i$. Because $T_i$ is
rightmost, $x$ and $y$ cannot both lie on the interior of the same arc
so we can assume that $x$ belongs to the upper arc and $y$ belongs to the
lower arc. The forwarding decision at $p_i$ is made as follows:
\begin{itemize}
\item If $p_i$ belongs to the upper arc, we walk clockwise
along $C_i$ until we reach vertex $x$. 
\item If $p_i$ belongs to the lower arc, we walk counterclockwise
along $C_i$ until we reach $y$. 
\end{itemize}
If $p_i = w_i$ we apply the first (upper arc) rule.

Once we reach $p_{i+1} = x$ or $y$, we repeat the process until we reach $t$.
Note that because the two vertices of $T_i$ other than $p_{i+1}$ are not both
below or both above line segment $[st]$, $T_i$ must be the leftmost triangle
that has $p_{i+1}$ as a vertex. Unless $p_{i+1} = t$, $p_{i+1}$ is a vertex of at
least another triangle intersecting $[st]$, so $T_i$ cannot be the rightmost
triangle that has $p_{i+1}$ as a vertex.

Figure~\ref{fig:exampleRouting} shows an example of a route computed by this
algorithm.

\begin{figure}[htb!]
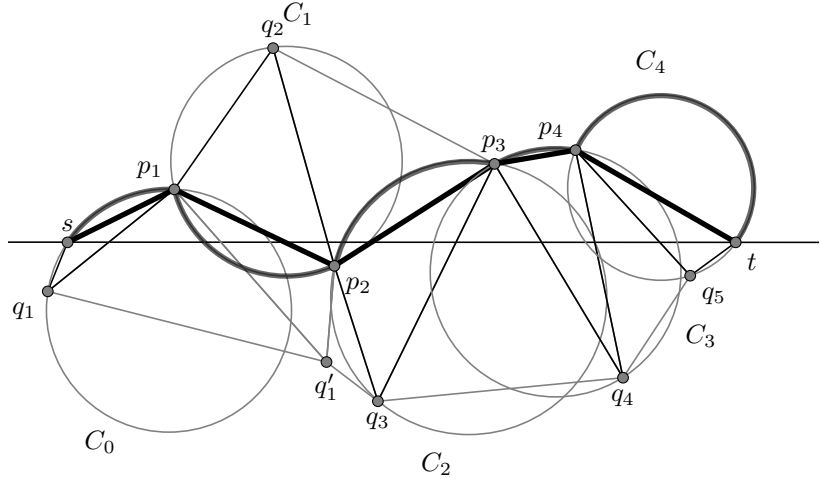

\center{\figexampleRouting}
\caption{Illustration of Chew's routing algorithm. The empty circles of the
\emph{rightmost triangles} are drawn in gray and their edges are drawn in black.
The edges of the obtained path and the associated arcs are thicker.}
\label{fig:exampleRouting}
\end{figure}

Because the routing decision can always be applied, because the decision
is based on the rightmost triangle and progress is made from left to right, 
and because $P$ is finite, we can conclude that the following results by Chew
from~\cite{Che86} extend to Delaunay
triangulations.  The following is Lemma~2 in~\cite{Che86}:
\begin{lemma}
\label{lem:Chew2}
The triangles used ($T_0,T_1\dots,T_k$) are ordered along
$[st]$. Although not all Delaunay triangulation triangles intersecting $[st]$
are used, those used appear in their order along $[st]$.
\end{lemma}
In Figure~\ref{fig:exampleRouting} the triangles $T_i$ are drawn with blue
edges. The following corollary is in~\cite{Che86} as well:
\begin{corollary}
The algorithm terminates, producing a path from $s$ to $t$.
\end{corollary}

\section{Routing Ratio}
\label{sec:compet}

In this section,
we prove the main theorem of this paper.
\begin{theorem}
\label{thm:main}
The Chew's routing algorithm on the Delaunay triangulation has a routing ratio of at most $(1.185043874 + 3\pi/2) \approx 5.89743256$.
\end{theorem}

As shown in Figure~\ref{fig:LBChew}, Chew's algorithm has routing ratio at least 5.7282 (see Section~\ref{sec:lower}).

%\begin{comment}
\begin{figure}[htb!]
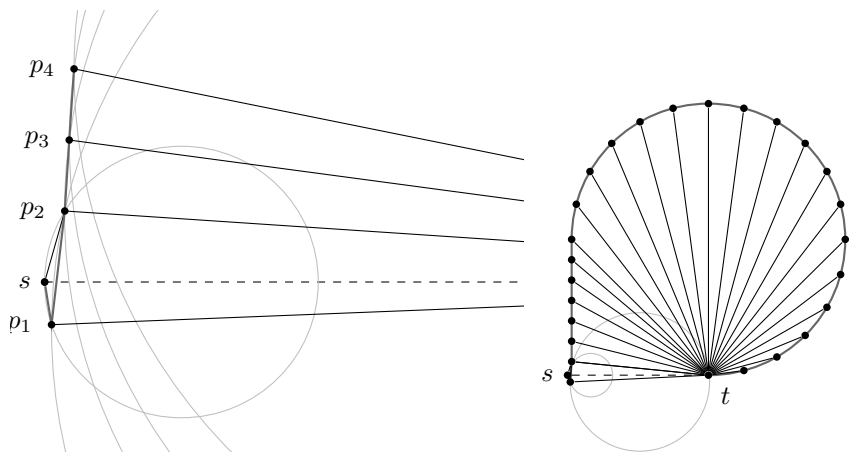

\center{\figLBChew}
%%routingMax.ggb and routingMaxZoom.ggb
\caption{On the right is a Delaunay triangulation that illustrates the lower
bound on the routing ratio of Chew's algorithm. The path obtained by the
algorithm is shown in dark gray; it has length $5.7282|[st]|$, a bit greater
than $(1+3\pi/2)|[st]|$. The left image zooms in on what happens close to
point $s$.}
\label{fig:LBChew}
\end{figure}
%\end{comment}

We devote this section to the proof of Theorem~\ref{thm:main}.

\subsection{Preliminaries}

We start by introducing additional definitions, notations, and structural
results about Chew's routing algorithm. Some of the notations are illustrated in
Figure~\ref{fig:example}.

We denote by $|[pq]|$ the Euclidean length of the line segment $[pq]$, and 
by $|\pathf{P}|$ the length of a path $\pathf{P}$ in the plane. Given a path
$\pathf{P}$ 
from $p$ to $q$ and a path $\pathf{Q}$ from $q$ to $r$, $\pathf{P}+\pathf{Q}$ 
denotes the concatenation of $\pathf{P}$ and $\pathf{Q}$. We say that the path $\pathf{P}$ from 
$p$ to $q$ is \emph{inside} a path $\pathf{Q}$ that also goes from $p$ to $q$ 
if the path $\pathf{P}$ is inside the region delimited by $\pathf{Q}+[qp]$. 
Note that if $\pathf{P}$ is convex and inside $\pathf{Q}$ then 
$|\pathf{P}| \leq |\pathf{Q}|$.
Given a path $\pathf{P}$ and two points $p$ and $q$ on $\pathf{P}$, we denote by
$\pathf{P}\subP{p,q}$ the sub-path of $\pathf{P}$ that goes from $p$ to $q$.

Let $s=p_0,p_1, \dots, p_k=t$ be the sequence of vertices visited by Chew's
routing algorithm. If some $p_i$ other than $s$ or $t$ lies on the
segment $[st]$, we can separately analyze the routing ratio of the
paths from $s$ to $p_i$ and from $p_i$ to $t$. We assume, therefore, that
no $p_i$, other than $s=p_0$ and $t=p_k$, lies on segment $[st]$.

For every edge $(p_i,p_{i+1})$, there is a corresponding 
oriented arc of $C_i$ used by the algorithm which we refer to as
$\dR{p_i,p_{i+1}}$ (shown in gray in Figures~\ref{fig:exampleRouting}
and~\ref{fig:example}). The orientation (clockwise or counterclockwise) of $\dR{p_i,p_{i+1}}$ is
the orientation taken by the routing algorithm when going from $p_i$ to
$p_{i+1}$. Let $\pathR$ be the union of these arcs. We call $\pathR$
the routing path from $s$ to $t$.
The length of the path $s=p_0,p_1, \dots,p_{k-1},p_k=t$ along the edges of the
Delaunay triangulation is smaller than the length of $\pathR$.
In the remainder of this section, we analyze the length of $\pathR$ 
to obtain an upper bound on the length of the path along the edges of the
Delaunay triangulation.

\subsection{Worst Case Circles $\Cpi{i}$}

In order to bound the length of $\pathR$, we work
with the circles $\Cpi{i}$ defined as follows. Let $A_1$ and $A_2$ be two circles
that go through $p_i$ and $p_{i+1}$ such that $A_1$ is tangent to $[st]$ 
and the tangent of $A_2$ at $p_i$ is vertical.
We define $\Cpi{i}$ to be $A_2$ if $A_2$ intersects $[st]$ twice and $A_1$
otherwise. 
Let $w'_i$ be the leftmost vertex of $\Cpi{i}$ and $O'_i$ the center of
$\Cpi{i}$. In the example of Figure~\ref{fig:example}, $w'_1 = p_1$ and
$w'_3 \neq p_3$.

Note that if $[p_{i}p_{i+1}]$ crosses
$[st]$, then
$\Cpi{i}$ is $A_2$. We consider three
types of circles $\Cpi{i}$:
\begin{itemize}
\item Type $A_1$: $p_i \neq w'_i$ and $[p_{i-1}p_i]$ does not cross $[st]$.
\item Type $A_2$: $p_{i}=w'_i$ and $[p_{i-1}p_i]$ does not cross $[st]$.
\item Type $B$: $[p_{i-1}p_i]$ crosses $[st]$.
\end{itemize}
In Figure~\ref{fig:example},  $\Cpi{0}, \Cpi{1} \dots \Cpi{4}$
are respectively of type $A_2$, $B$, $B$, $A_1$, $A_2$. Note that if $\Cpi{i}$
is of type $B$, then $p_i = w'_i$. We use the expression ``type $A$'' instead of ``type $A_1$ or $A_2$''.

Given two points $p,q$ on $\Cpi{i}$, let $\arcP{i}{p,q}$ be
the arc on $\Cpi{i}$ from $p$ to $q$ whose orientation (clockwise or
counterclockwise) is the same as the orientation of $\dR{p_i,p_{i+1}}$ around
$C_i$.
Notice that $|\dR{p_i,p_{i+1}}| \leq |\arcP{i}{p_i,p_{i+1}}|$.

Let $t_i$ be the first point $p_j$ after $p_i$ such that $[p_ip_j]$
intersects $st$. Notice that $t_{k-1} = t$. We also set $t_k=t$.
In Figure~\ref{fig:example}, 
$t_0=p_1$, $t_1=p_2$, $t_2 = p_3$ and $t_3=t_4=t_5=t$.
In addition,
let $t'_i=(x(t_i),0)$ and $s_i = (x(w'_i), 0)$.

\begin{lemma}
\label{lem:siOrder}
For all $0 < i \leq k$:
\begin{align}
x(s_{i-1}) \leq x(s_{i}),\label{eq:siOrder}\\
x(s_{i}) \leq x(t_{i-1}) \leq x(t_{i}).\label{eq:tiOrder}
\end{align}
\end{lemma}

\begin{proof}
We first prove (\ref{eq:siOrder}). If $p_{i}=w'_{i}$ then 
$x(s_{i}) = x(p_{i})$. Since $p_{i}$ lies on $\Cpi{i-1}$ and
$x(s_{i-1})$ is the $x$-coordinate of the leftmost point of $\Cpi{i-1}$, we have
that $x(s_{i-1})\leq x(p_{i})=x(s_{i})$. If $p_{i} \not= w'_{i}$ then
$\Cpi{i}$ is of type $A_1$. We assume without loss of generality that $p_i$ lies above $st$.
If $w'_{i}$ is on or in the interior of $\Cpi{i-1}$ then (\ref{eq:siOrder})
holds. Otherwise, the
rightmost intersection of $\Cpi{i-1}$ with $st$ must be to the left of the
intersection of $w'_ip_i$ and $st$. This, in turn, implies
(\ref{eq:siOrder})

We now prove (\ref{eq:tiOrder}). We first observe that $t_{i-1}=p_j$ and
$t_{i}=p_{j'}$ for some $i \leq j \leq j'$. Using  inequality
(\ref{eq:siOrder}), we have that $x(s_i) \leq x(s_j) \leq x(p_j) = x(t_{i-1})$,
so the first inequality in (\ref{eq:tiOrder}) holds. 
The second inequality trivially holds when $j=j'$, so we assume otherwise.
In that case, $p_j, p_{j+1}, \dots, p_{j'-1}$ must all be on the same side of
$st$. Without loss of generality, we assume that
$p_{j'}$ lies above $st$. This implies that
$[p_{j'-1}p_{j'}]$ crosses $[st]$ and therefore $\Cpi{j-1}$ is of type $B$.
Moreover,
$p_{j'-1}=w'_{j'-1}$ and $p_{j'-1}$ lies below $st$.
On the other hand, $p_j$ lies below $st$ and on $\Cpi{j-1}$, which is of
type $B$ and whose
center $O'_{j-1}$ is above $st$. Note that $p_{j'-1}$ and
$p_{j'}$ lie outside of $\Cpi{j-1}$ and that 
$x(w'_{j-1}) \leq x(w'_{j'-1})$ (recall that $w'_{j'-1} = p_{j'-1})$. 
Therefore, if $q$ is the intersection of $\arcP{j'-1}{p_{j'-1},p_{j'}}$ and $st$, no point
of $\Cpi{j-1}$ below $st$ and outside of $\Cpi{j'-1}$ has an $x$-coordinate larger
than $x(q)$. Since $x(q) < x(p_{j'})$, the second inequality in
(\ref{eq:tiOrder}) holds.
\end{proof}

\begin{figure}
\center{\figexample}
\caption{Illustration of Lemma~\ref{lem:induction} on the example of Figure~\ref{fig:exampleRouting}. The empty circles of the Delaunay triangulation and the routing path $\pathf{R}$ are drawn in gray. Worst cast circles $\Cpi{i}$, paths $\dP{i}$, and segments of height $|y(t_i)|$ are drawn in black. Lengths $|[t'_{i-1}t'_i]|$ are represented by dashed horizontal segments.}
\label{fig:example}
\end{figure}

\subsection{Proof of Theorem~\ref{thm:main}}

In this section,
we prove our main theorem.
Given two points $p$ and $q$ such that $x(p) < x(q)$ and $y(p) = y(q)$,
we define the path $\snailP{p,q}$ as follows. Let $C$ be the circle above
$pq$ that is tangent to $pq$ at $q$ and tangent to the line $x = x(p)$ at a point
that we denote by $p'$. The path $\snailP{p,q}$ consists of $[pp']$ together with
the clockwise arc from $p'$ to $q$ on $C$. We call $\snailP{p,q}$ the
\emph{snail curve}
from $p$ to $q$. Note that $|\snailP{p,q}| = (1+3\pi/2)(x(q)-x(p))$.
We also define the path $\dP{i}$ to be
$[s_iw'_i] + \arcP{i}{w'_{i}, p_{i+1}}$ for $0\leq i \leq k-1$ (see
Figure~\ref{fig:example}). 

We start with a lemma that motivates these definitions.
\begin{lemma}
\label{lem:warmup}
$|\dP{k-1}| \leq |\snailP{s_{k-1},t}|$ 
\end{lemma}

\begin{proof}
This follows from the fact that $\dP{k-1}$ from $s_{k-1}$ to
$t$ is convex and inside $\snailP{s_{k-1},t}$.
\end{proof}

The following lemma is the key to proving Theorem~\ref{thm:main}.
\begin{lemma}
\label{lem:induction}
For all $0 < i < k$ and $\delta = \deltaMin$, 
\begin{equation}
\label{eq:induction2}
|\dP{i-1}| + |y(t_{i-1})| \leq |\dP{i}\subP{s_i,p_i}| + |\snailP{s_{i-1},s_i}| + |y(t_i)| + \delta |[t'_{i-1}t'_{i}]|.
\end{equation}
Moreover, if $\Cpi{i-1}$ is of type $A$ (then $t_{i-1}=t_{i}$), the previous inequality is equivalent to
\begin{equation}
\label{eq:induction1}
|\dP{i-1}| \leq |\dP{i}\subP{s_i,p_i}| + |\snailP{s_{i-1},s_i}|.
\end{equation}
\end{lemma}
This lemma is illustrated in Figure~\ref{fig:example}.
We first show how to use Lemma~\ref{lem:induction} to prove
Theorem~\ref{thm:main},
then we prove Lemma~\ref{lem:induction} in Section~\ref{subsection proof key lemma}.\\
\begin{proofof}{\bf Theorem~\ref{thm:main}.}
By Lemma~\ref{lem:siOrder}, $\sum_{i=1}^{k-1} |[t'_{i-1}t'_i]| < |[st]|$
and $\sum_{i=1}^{k} |\snailP{s_{i-1},s_i}| = |\snailP{s,t}|$. By summing the  $k-1$ inequalities from Lemma~\ref{lem:induction} and the
inequality from Lemma~\ref{lem:warmup}, we get
$$\sum_{i=1}^{k} |\dP{i-1}| + |y(t_0)| < \sum_{i=1}^{k-1} |\dP{i}\subP{s_i,p_i}| + |\snailP{s,t}| + |y(t_{k-1})| + \delta |[st]|.$$
The fact that $t_{k-1} = t$ implies $y(t_{k-1}) = 0$.
Therefore,
since $\dP{i-1} = \arcP{i-1}{p_{i-1}p_i} + \dP{i-1}\subP{s_{i-1},p_{i-1}}$,
we have
$$|\pathR| \leq \sum_{i=1}^{k} \arcP{i-1}{p_{i-1}p_i} < |\snailP{s,t}| + \delta |[st]| \leq  (1.185043874 + 3\pi/2) |[st]|$$
which completes the proof.
\end{proofof}

\subsection{Proof of the Key Lemma}
\label{subsection proof key lemma}

In this section,
we prove Lemma~\ref{lem:induction}. We will make use of the following lemma
whose proof is in the Appendix.\\

\begin{lemma}
\label{lem:alpha}
Let $\theta = \angle w'_{i-1}O'_{i-1}p_i$ and
$\alpha = \angle w'_iO'_{i}p_i$ be the angles defined using the orientations
$\arcP{i}{p_{i-1},p_{i}}$ and $\arcP{i}{p_i,p_{i+1}}$, respectively. Then
\begin{equation}
\label{eq:anglesp}
0 \leq \alpha < \theta < 3\pi/2.
\end{equation}
\end{lemma}

\begin{proofof}{\bf Lemma~\ref{lem:induction}.}
Notice that if $\Cpi{i-1}$ is of type $A$,
then $|y(t_{i-1})| = |y(t_i)|$.
Hence,
in this case,
it sufficient to prove
\begin{equation}
\label{eq:induction1-bis}
|\dP{i-1}| \leq |\dP{i}\subP{s_i,p_i}| + |\snailP{s_{i-1},s_i}| .
\end{equation}
For the rest of the proof,
we consider three cases depending on the types of $\Cpi{i-1}$ and
$\Cpi{i}$.

%%%%%%%%%%%%%%%%%%%%%%%%%%%%%%%%%%%%%%
\noindent $\bullet$ {\bf $\Cpi{i-1}$
is of type $A$ and $\Cpi{i}$ is of type $A_2$ or $B$}.
In this case,
$p_i = w'_i$ from which $x(s_i) = x(p_i)$ follows.
Let $X$ be the orthogonal projection
of $p_i$ onto $s_{i-1}w'_{i-1}$. Then
\begin{equation}\label{eq:flip}
|[s_{i-1},X]+\snailP{X,p_i}| =  |\dP{i}\subP{s_i,p_i}| + |\snailP{s_{i-1},s_i}|.
\end{equation}
Since the path $\dP{i-1}$ is convex and inside the path 
$[s_{i-1},X]+\snailP{X,p_i}$
(see Figure~\ref{fig:caseA1B}),
\begin{figure}[hbt!]
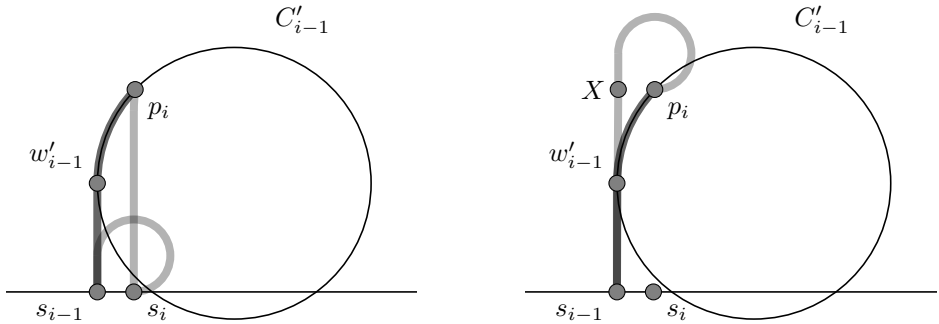

  \center{\figcaseAOneB}
  \caption{Illustration of the case when $\Cpi{i-1}$
is of type $A$ and $\Cpi{i}$ is of type $A_2$ or $B$.}
  \label{fig:caseA1B}
\end{figure}
we get
$|\dP{i-1}| \leq |[s_{i-1},X]+\snailP{X,p_i}|$. 
Applying Inequality (\ref{eq:flip}) completes the proof of this case.

%%%%%%%%%%%%%%%%%%%%%%%%%%%%%%%%%%%%%%
\noindent $\bullet$ {\bf $\Cpi{i-1}$ is of type $A$ and $\Cpi{i}$ is of type $A_1$.}
(See Figure~\ref{fig:caseAA1}.) 
Let $b_i$ the lowest point of $\Cpi{i}$.
%Observe that if $\Cpi{i}$ is of type $A_1$, $s''_i=s_i$.
Let $X$ and $Y$ be respectively the projections of $s_{i-1}$ and $b_i$ on the line $y=y(b_{i-1})$.
We consider the snail curve $\snailP{s_{i},b_i} = \dP{i}\subP{s_i,p_i} + \arcP{i}{p_{i},b_i}$
($\arcP{i}{p_{i},b_i}$ see Figure~\ref{fig:caseAA1}).
\begin{figure}[hbt!]
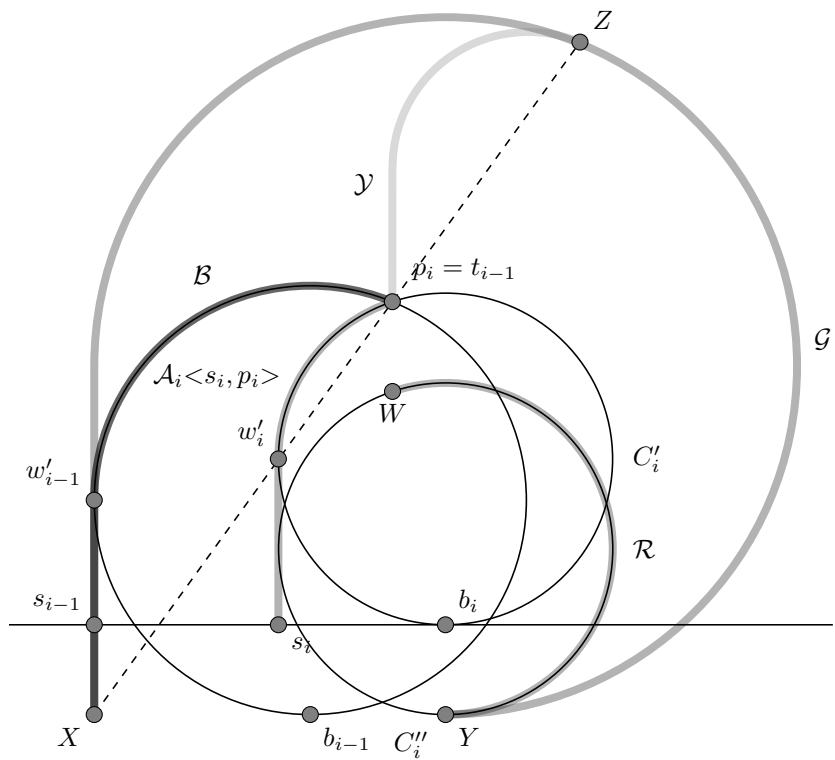

\center{\figcaseAAOne}
  \caption{Illustration of the case when $\Cpi{i-1}$ is of type $A$ and $\Cpi{i}$ is of type $A_1$.}
  \label{fig:caseAA1}
\end{figure}

Let $\pathf{G} = \snailP{X,Y}$. We have:
\begin{equation}
\label{eq:Gp}
|\dP{i}\subP{s_i,p_i}| + |\snailP{s_{i-1},s_i}| + |\arcP{i}{p_{i},b_i}| =  |\snailP{s_{i-1},b_i}| =|\pathf{G}|.
\end{equation}
Let $\pathf{B}=[Xs_{i-1}]+\dP{i-1}$
and $Z\neq X$ be the intersection of $X p_i$ with $\pathf{G}$. Denote by $\pathf{Y}$ the path from $p_i$ to $Z$ that is homothetic to $\pathf{B}$. Note that $\pathf{B}$ and $\pathf{Y}$ are both homothetic to $\pathf{G}\subP{X,Z}$.
Hence,
\begin{equation}\label{eq:homoY}
|\pathf{B} + \pathf{Y}| = |\pathf{G}\subP{X,Z}|.
\end{equation}

Let $C''_i$ be the curve obtained by translating down $\Cpi{i}$ until $b_i$ lies on $Y$.
Denote by $\pathf{R}$ and $W$ the images of $\arcP{i}{p_{i},b_i}$ and $p_i$ by the same translation, respectively. The circle $C''_i$ is tangent to the circular section of $\snailP{s_{i-1},b_i}$ at $Y$. Moreover,
the radius of $C''_i$ is smaller than the radius of the circular section of $\snailP{s_{i-1},b_i}$. Hence, it does not intersect $\snailP{s_{i-1},b_i}$. This implies that if $\pathf{R}$ intersects the line $p_iZ$, the intersection points must be in $[p_iZ]$. We can show that $\pathf{R}$ does not intersect $\pathf{Y}$.
Since $\arcP{i}{p_{i},b_i}$ is convex and lies inside $[Wp_i]+\pathf{Y}+\pathf{G}\subP{Z,Y}$,
we have
\begin{equation}\label{eq:insideG1}
|\pathf{R}| \leq |[Wp_i]+\pathf{Y}+\pathf{G}\subP{Z,Y}|.
\end{equation}

Summing (\ref{eq:homoY}) and (\ref{eq:insideG1}),
and removing $|\pathf{Y}|$ and $|[b_iY]|=|[Xs_{i-1}]|$ from both sides,
we get $|\dP{i-1}| + |\arcP{i}{p_{i},b_i}| \leq |\pathf{G}|$.
Using (\ref{eq:Gp}) and removing $|\arcP{i}{p_{i},b_i}|$ from both sides,
we get (\ref{eq:induction1-bis}).

%%%%%%%%%%%%%%%%%%%%%%%%%%%%%%%%%%%%%%
\noindent $\bullet$ {\bf $\Cpi{i-1}$ is of type $B$ and $\Cpi{i}$ is of type $A$ or $B$}.
In this case,
$w'_{i-1} = p_{i-1}$ and
$t_{i-1} = p_i \not= t_i$. 
We consider the
case where $|y(p_{i-1})-y(p_{i})| < |y(p_{i-1})-y(t_i)|$ first
(refer to Figure~\ref{fig:caseBall2}).
\begin{figure}[hbt!]
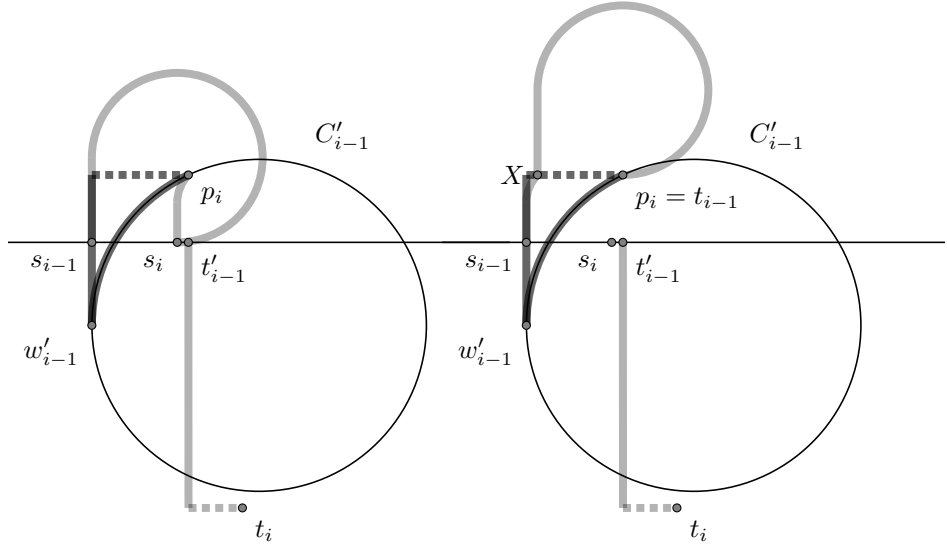

\center{\figcaseBallTwo}
  \caption{Illustration of the case when $\Cpi{i-1}$ is of type $B$ and $|y(p_{i-1})-y(p_{i})| < |y(p_{i-1})-y(t_i)|$.}
  \label{fig:caseBall2}
\end{figure}
Recall that 
$\dP{i-1} = \arcP{i-1}{p_{i-1}p_i} + [s_{i-1}p_{i-1}]$ and let ${\cal P^*}$ be the curve obtained by translating $\dP{i}\subP{s_i,p_i}$ to the left until $s_i$ lies on $s_{i-1}$.
Denote the highest point of ${\cal P^*}$ by $X$.
Notice that $x(p_i) - x(X) = x(s_i) - x(s_{i-1})$, $\arcP{i-1}{p_{i-1}p_i}$ is convex
and $\arcP{i-1}{p_{i-1}p_i}$ is inside $[p_{i-1}s_{i-1}] + {\cal P^*} + \snailP{X, p_i}$, which is also convex.
Consequently, we get
\begin{equation}
\label{eq:halfTi1}
|\arcP{i-1}{p_{i-1}p_i}| \leq |[p_{i-1}s_{i-1}]| + |\snailP{s_{i-1},s_i}| + |\dP{i}\subP{s_i,p_i}|.
\end{equation}
Moreover, since  $|y(p_{i-1})-y(p_{i})| < |y(p_{i-1})-y(t_i)|$, we have
\begin{equation}
\label{eq:halfTi2}
|[s_{i-1}p_{i-1}]| + y(t_{i-1}) \leq |y(t_i)|-|[s_{i-1}p_{i-1}]| .
\end{equation}
Summing (\ref{eq:halfTi1}) and (\ref{eq:halfTi2}),
we get
$$|\dP{i-1}| + |y(t_{i-1})| \leq |\dP{i}\subP{s_i,p_i}| + |\snailP{s_{i-1},s_i}| + |y(t_i)|.$$
Notice that we did not need the additional potential
$\delta|[t'_{i-1}t'_i]|$ in this case.

For the rest of the proof,
we can assume that $|y(p_{i-1})-y(p_{i})| \geq |y(p_{i-1})-y(t_i)|$.
If we assume that $p_i$ lies above $st$, then $t_i$ must lie below $st$.
The point $t_i$ is outside of $C_{i-1}$.
By Lemma~\ref{lem:siOrder},
$x(p_i) = x(t_{i-1}) < x(t_i)$.
Moreover,
all points $p$ on $\Cpi{i-1}$ or inside of it, and such that $x(p_i) < x(p)$
are in the interior of $C_{i-1}$.
Therefore,
$t_i$ is outside of $\Cpi{i-1}$.

Recall that by Lemma~\ref{lem:alpha}, if $\theta = \angle w'_{i-1}O'_{i-1}p_i$
and $\alpha = \angle w'_iO'_{i}p_i$, then $0 \leq \alpha < \theta < 3\pi/2$.
Without loss of generality, assume that the radius of $\Cpi{i-1}$ is $1$
and the radius of $\Cpi{i}$ is $R$. Then we have
$|\dP{i}\subP{s_i,p_i}| = (1 + \alpha)R$ and 
$[s_it'_{i-1}] = (1-\cos(\alpha))R$. Let $D$ be the difference between
the left-hand side and the right-hand side of inequality (\ref{eq:induction2}). We can write $D$ as
\begin{align*}
\label{eq:worstcase}
 D &= |\snailP{s_{i-1},s_{i}}| +|\dP{i}\subP{s_i,p_i}| + \delta
|[t'_{i-1}t'_{i}]| + |y(t_i)| - |\dP{i-1}| - |y(t_{i-1})| \\
 &=  (1+3\pi/2)(1-\cos(\theta) - (1-\cos(\alpha))R) + (1+\alpha)R +\delta |[t'_{i-1}t'_{i}]|+ |y(t_i)|  -\theta - \sin(\theta).
\end{align*}
It remains to prove that $D>0$.

We first consider the case where $\theta \leq \pi/4$,
which implies that $\alpha < \pi/4$ as well.
Let $p'_{i} \neq p_i$ be the intersection of
$\Cpi{i-1}$ with
the horizontal line through $p_i$.
Since $\theta \leq \pi/4$,
we have $x(p'_i)> x(O'_i)$.

Since $|y(p_{i-1})-y(p_{i})| < |y(p_{i-1})-y(p)|$
for all points $p$ outside of $\Cpi{i-1}$ such that $x(p_i) \leq x(p) \leq x(p'_i)$,
it follows that $x(t_i) > x(p'_i)$.
Note that $\angle w'_{i-1}O'_{i-1}p'_{i} = \pi - \theta$,
as illustrated in Figure~\ref{fig:caseBall}. 
\begin{figure}[hbt!]
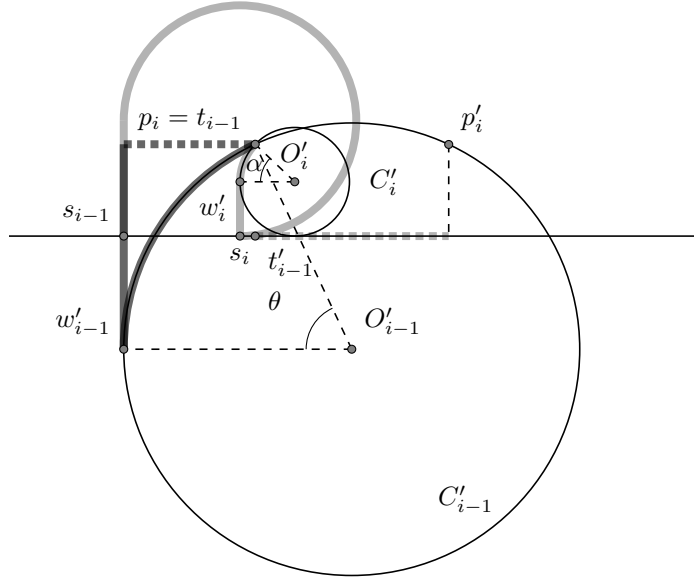

\center{\figcaseBall}
  \caption{Illustration of the case of when $\Cpi{i-1}$ is of type $B$  and $|y(p_{i-1})-y(p_{i})| \geq |y(p_{i-1})-y(t_i)|$.}
  \label{fig:caseBall}
\end{figure}
Since $|[t'_{i-1}t'_i]| > |[t_{i-1}p'_{i}]| = 2 \cos(\theta)$
(recall that $p_i = t_{i-1}$),
we have
\begin{align*}
D &\geq  (1+3\pi/2)(1-\cos(\theta) - (1-\cos(\alpha))R) + (1+\alpha)R +2\delta\cos(\theta) -\theta - \sin(\theta) \\
&\geq  R[1+\alpha-(1+3\pi/2)(1-\cos(\alpha))] + (1+3\pi/2)(1-\cos(\theta)) +2\delta\cos(\theta) -\theta - \sin(\theta).
\end{align*}
There exists an $\alpha_0>\pi/4$ such that $1 + \alpha - (1+3\pi/2)(1-\cos(\alpha))$ is positive for all $\alpha\in[0,\alpha_0]$.
Therefore,
to prove that $D > 0$ (and therefore that inequality (\ref{eq:induction2}) holds),
it is sufficient to prove that
\begin{align*}
  0 &\leq (1+3\pi/2)(1-\cos(\theta)) + 2\delta\cos(\theta) - \theta - \sin(\theta).
\end{align*}
If we take $\delta = \deltaMin$,
we can show that this inequality is true using elementary calculus arguments.

To complete the proof,
it remains to consider the case where $\theta\in [\pi/4,\pi]$.
If $\alpha\leq \alpha_0$,
we have $D \geq (1+3\pi/2)(1-\cos(\theta)) - \theta - \sin(\theta)$,
which is positive for all $\theta\in [\pi/4,\pi]$.
If $\alpha \in [\alpha_0,\pi]$, $1 + \alpha - (1+3\pi/2)(1-\cos(\alpha))$
is negative and decreasing.
Thus,
since $\alpha\leq \theta$ and $R<1$,
we obtain
\begin{align*}
D &\geq  1+\theta-(1+3\pi/2)(1-\cos(\theta)) + (1+3\pi/2)(1-\cos(\theta)) +\delta |[t'_{i-1}t'_{i}]|+ |y(t_i)|  -\theta - \sin(\theta)\\
 &\geq 1-\sin(\theta)+\delta |[t'_{i-1}t'_{i}]|+ |y(t_i)| .
\end{align*}
This lower bound is trivially positive,
hence inequality (\ref{eq:induction2}) holds in all cases.
\end{proofof}

\section{Lower Bounds}
\label{sec:lower}
In this section, we provide lower bounds on the routing ratio and the competitive ratio of any $k$-local routing algorithm on the Delaunay triangulation and the $L_1$- and $L_\infty$-Delaunay triangulation.

\begin{theorem}
\label{thm:ChewLower}
The routing ratio of Chew's routing algorithm on Delaunay triangulations is at least $5.7282$.
\end{theorem}
\begin{proof}
  Let $C_0$ and $C_1$ be two circles such that the west point of $C_0$ lies on the $x$-axis and the west point $w_1$ of $C_1$ lies on $C_0$ and below the $x$-axis. Let $s$ be the west point of $C_0$ and let $t$ the rightmost intersection of $C_1$ with the $x$-axis. Let $p_1 = w_1$ and $p_2$ be the intersections of $C_0$ and $C_1$. We perturb the configuration such that $s$ lies slightly below the $x$-axis and $p_1$ lies slightly above the horizontal line through $w_1$ (see Figure~\ref{fig:LBChew}). This implies that the first two edges of the path computed by Chew's algorithm are $[s p_1]$ and $[p_1 p_2]$. 

  Next, we add circles $C_i$ with west point $w_i$ such that $t$ lies on $C_i$, $p_i$ lies slightly above $w_i$, and point $p_{i+1}$ lies slightly above $p_i$. We place circles until $t$ is the lowest point of $C_j$ for some $j$. Finally, we add points $p_j, p_{j+1}, ..., p_k$ (for some integer $k$) on the $\arcP{j}{p_j t}$. We slightly perturb the configuration such that all chords reach $t$ (see Figure~\ref{fig:LBChew}). Observe that by placing sufficiently many vertices between $p_2$ and $p_j$, we create an almost vertical path from $p_2$ to $p_j$. The routing path computed by Chew's algorithm tends to $[s p_1] + [p_1 p_2] + \snailP{p_2,t}$. 
  
  We now pick $C_0$ to be the circle with center at $O_0 = (-0.7652277146, 0)$ and radius $0.2369448832$ and we pick $C_1$ to be the circle with center $O_1 = (0, -0.0320133045)$ and radius 1. This leads to a routing path whose length approaches $11.4660626$ as $j$ and $k$ approach infinity. Since the distance between $s$ and $t$ is $2.00166$, this implies that the routing ratio of Chew's algorithm is at least $5.7282$.\footnote{GeoGebra files that describe the 2 first examples presented in this section are available at the following url: \url{http://www.labri.fr/perso/bonichon/DelaunayRouting/}}
\end{proof}

We note that this lower bound is strictly larger that $|\snailP{s,t}|/|[st]| = 1+3\pi/2$. Next we show that no deterministic $k$-local routing algorithm on Delaunay triangulations can have routing ratio less than 1.7018. 

\begin{theorem}
\label{thm:L2Lower}
There exist no deterministic $k$-local routing algorithm on Delaunay triangulations with routing ratio at most 1.7018 or that is $1.2327$-competitive.
\end{theorem}
\begin{proof}
  Let $C_0$ be the circle with center $O_0 = (1, 0)$ and radius $1$ and let $C_1$ be the circle with center $O_1=(1.4804533538, 0.2990071425)$ and radius $1.2285346394$. Let $s$ be the leftmost intersection of $C_0$ with the $x$-axis and let $t$ be the rightmost intersection of $C_1$ with the $x$-axis. Let points $A$ and $B$ be the intersections of $C_0$ and $C_1$ such that $A$ lies above the $x$-axis. Let $A'$ and $B'$ be the vertices on $O_0 A$ and $O_0 B$ outside $C_0$ such that $|[AA']| = |[BB']|= 0.0718725166$ (see Figure~\ref{fig:LBL2}). These points (referred to as \emph{shield vertices} by Bose~et~al.~\cite{bose2011almost}) will ensure that no shortcuts between points on $C_0$ and points on $C_1$ are created. 

\begin{figure}[htb!]
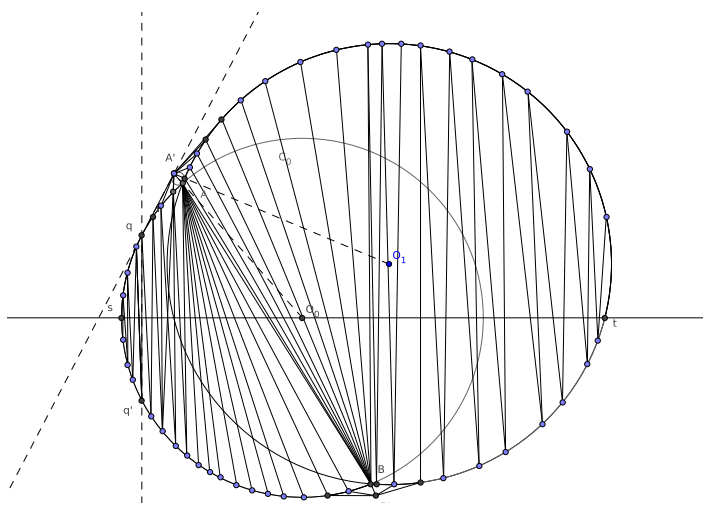

\center{\figLBLTwo}
%%2circlesFullCompet.ggb
\caption{One of the two point sets used to provide a lower bound on the routing ratio on Delaunay triangulations.}
\label{fig:LBL2}
\end{figure}

  Next, we place points densely on the arcs of $C_0$ and $C_1$ that are not contained in the other circle. To ensure that $A A'$ and $B B'$ are edges of the Delaunay triangulation, we leave small gaps along the arcs close to each shield vertex. Since all points of $C_0$ (respectively $C_1$) are co-circular, any planar triangulations of them is a valid Delaunay triangulation. Next, we perturb the points in order to both break co-circularity and to choose the triangulation of the interior of the circles (see Figure~\ref{fig:LBL2}). We compute a triangulation where the shortest paths between $t$ and any point of point set $P$ does not use any chord except $[A B]$. Let $q$ be the point on $C_0$ such that $q A'$ is the tangent of $C_0$ at $q$ and let $q'$ the reflection of $q$ over the $x$-axis.

  Now let us consider any deterministic $k$-local routing algorithm. We consider two point sets: The first one is described above (see Figure~\ref{fig:LBL2}) and the second one is obtained from the first one by reflecting the part of the point set that lies to the right of $q q'$ over the $x$-axis. No deterministic $k$-local routing algorithm can distinguish between the two instances before it reaches $q''$, a $k$-neighbor of $q$ or $q'$, depending on if the routing algorithm followed the arc towards $q$ or $q'$. Since vertices are densely placed on the arc of $C_0$, $q''$ is arbitrary close to $q$ or $q'$. Hence, any deterministic $k$-local algorithm must route to the same vertex $q''$ in both instances. Either $q''$ is close to $q$ or $q'$. In one of the two instances this leads to a non-optimal route: On the instance of Figure~\ref{fig:LBL2}, the length of the arc from $s$ to $q$ is $0.477998$ and the shortest paths from $q$ to $t$ go via vertex $A$ and have length $4.0693551467$. The Euclidean distance from $s$ to $t$ is $2.6720456033$. Hence, on one of the two instances the length of the computed path is at least $1.7018 \cdot |[st]|$. The shortest path between $s$ and $t$ in this configuration is of length $3.6888$, this configuration gives a lower bound on the competitiveness of any routing algorithm of $1.2327$.
\end{proof}

Finally, we look at the routing ratio and competitiveness of any deterministic $k$-local routing algorithm on the $L_1$- and $L_\infty$-Delaunay triangulations. 

\begin{theorem}
\label{thm:L1Lower}
There exists no deterministic $k$-local routing algorithm for the $L_1$- and $L_\infty$-Delaunay triangulations that has routing ratio less than $(2 + \sqrt 2/2)\approx 2.7071$ or that is $\frac{2 + \sqrt 2/2}{1+\sqrt 2}\approx 1.1213$-competitive.
\end{theorem}
\begin{proof}
  The proofs for the $L_1$- and $L_\infty$-Delaunay triangulations are very similar as one can be created from the other by rotating the point set. We present only the point set of the $L_\infty$-Delaunay triangulation. First, we place the source vertex $s$ at the origin. Given some values $\epsilon>0$ and $d<1$, we then place $k$ vertices close to point $q=(\epsilon, (2-\sqrt 2)/4)$ and $q'=(2\epsilon,-(2-\sqrt 2)/4)$ (see Figure~\ref{fig:LBL1}). Next, we place a vertex A at $(3\epsilon,1-d+\epsilon)$, vertex $B$ at $(1+2\epsilon,1-d)$, and destination $t$ at $(1+3\epsilon,0)$. Finally, we place vertices densely on $[Bt]$ and on $[q'B']$, where $B'$ is picked such that $q',B,t,B'$ forms a parallelogram. As $\epsilon$ approaches 0, the resulting $L_\infty$-Delaunay triangulation approaches the one shown in Figure~\ref{fig:LBL1}. 

\begin{figure}[htb!]
\center{\figLBLOne}
%% L1LB.ggb
\caption{One of the two point sets used to provide a lower bound on the routing ratio on $L_\infty$-Delaunay triangulations.}
\label{fig:LBL1}
\end{figure}

  Now, consider any deterministic $k$-local routing algorithm. We consider two point sets: The first one is described above (see Figure~\ref{fig:LBL1}) and the second one is obtained from the first one by reflecting the part of the point set that lies to the right of $q q'$ over the $x$-axis. Since there are $k$ points between $s$ and $q$ and between $s$ and $q'$, the only information the $k$-local routing algorithm has before getting close to $q$ or $q'$ consists of the vertices to the left of $qq'$. If the first step made by the algorithm is towards a vertex close to $q$, we consider the point set shown in Figure~\ref{fig:LBL1}. Otherwise, we consider the reflected point set instead. We note that $|[sq]|= (2-\sqrt 2)/4$ and that the shortest paths from $q$ to $t$ have length $\min(|[qA]+[AB]+[Bt]|,|[qq']+[q'B]+[Bt]|)=\min(1-d-(2-\sqrt 2)/4+1+1-d,(2-\sqrt 2)/2+\sqrt 2 + 1-d)$. If we pick $d=(2-\sqrt 2)/4)$, the length of both paths is equal to $3-3(2-\sqrt 2)/4$. This leads to a path from $s$ to $t$ of length $2 + \sqrt 2/2$. Since the Euclidean distance between $s$ and $t$ approaches 1 as $\epsilon$ approaches 0, this gives a lower bound on the routing ratio of any deterministic $k$-local routing algorithm on the $L_1$- and $L_\infty$-Delaunay triangulations.

Finally, we observe that on the point set shown in Figure~\ref{fig:LBL1}, the length of shortest path from $s$ to $t$ is $1+\sqrt 2$. This gives a lower bound of $\frac{2 + \sqrt 2/2}{1+\sqrt 2}$ on the competitive ratio of any deterministic $k$-local routing algorithm.
\end{proof}

\bibliographystyle{abbrv}
\bibliography{bib}

\newpage
\appendix
\section*{Appendix}

\begin{figure}[hbt!]
  \center{\figangles}
  \caption{Illustration of Lemmas~\ref{lem:angles} and~\ref{lem:alpha}.}
  \label{fig:angles}
\end{figure}

\begin{lemma}
\label{lem:angles}
Let $O_i$ be the center of $C_i$ and let $\angle w_{i-1}O_{i-1}p_i$
and $\angle w_iO_ip_i$ be the angles defined using orientations of
$\dR{p_{i-1},p_{i}}$ and $\dR{p_i,p_{i+1}}$, respectively. Then (as illustrated
in Figure~\ref{fig:angles}), for every $1\leq i \leq k$,
$$0 \leq \angle w_iO_ip_i < \angle w_{i-1}O_{i-1}p_i < 3\pi/2.$$
\end{lemma}

\begin{proof}
Without loss of generality, we assume that $p_{i-1}$ and $p_i$ belong
to the upper arc from
$w_{i-1}$ to $r_{i-1}$ of $C_{i-1}$. This implies that the orientation of
$\dR{p_{i-1},p_{i}}$ is clockwise and that $p_i$ lies above $st$.
Because $r_{i-1}$ is the rightmost intersection of $C_{i-1}$ with $st$,
it follows that $\angle w_{i-1}O_{i-1}r_{i-1} < 3\pi/2$, which implies the
third inequality. Also, the first inequality holds because the angles are
defined to be non-negative.

Consider the triangles which intersect $[st]$ that have $p_i$ as a vertex.
Let $C_{i-1}=C_i^0, C_i^1, \dots, C_i^l=C_i$ be the circles circumscribing these triangles,
ordered from left to right.
Let $O_i^j$ be the center and $w_i^j$ be the leftmost
point of circle $C_i^j$. The lemma will follow if we show that
$\angle w_i^jO_i^jp_i < \angle w_i^{j-1}O_i^{j-1}p_i$ for every
$j = 1, \dots, l$. Let $p_i$, $q_j$ and $q'_{j}$ be the vertices of the triangle
circumscribed by $C_i^{j-1}$, in clockwise
order. Then $q_j$ must lie below the $x$-axis and $O_i^j$ must lie on 
the perpendicular bisector of segment $[p_iq_j]$ and in the half-plane
defined by $p_iO_i^{j-1}$ not including $w_i^{j-1}$.
This implies that $\angle p_iO_i^{j-1}w_i^{j-1} > \angle p_iO_i^jw_i^j$.
\end{proof}

\begin{proofof}{\bf Lemma~\ref{lem:alpha}}
This lemma is illustrated in Figure~\ref{fig:angles}.
Consider the circles $C_{i-1}$ and $C_i$ with centers $O_{i-1}$ and $O_i$, respectively. Without
loss of generality,
we assume that  $p_{i-1}$ and $p_i$ belong to the upper arc from $w_{i-1}$ to
$r_{i-1}$ of $C_{i-1}$. This implies that the orientation of $\dR{p_{i-1},p_{i}}$
is clockwise and that $y(p_i) \geq 0$.
The center $O'_{i-1}$ of $\Cpi{i-1}$ lies on the half-line defined
by the perpendicular bisector of $[p_{i-1}p_{i}]$, starting at
$O_{i-1}$ and intersecting $[p_{i-1}p_{i}]$. This implies that
a) $\angle w'_{i-1}O'_{i-1}p_{i} > \angle w_{i-1}O_{i-1}p_{i}$ and
b) $\angle w'_{i-1}O'_{i-1}p_{i-1} < \angle w_{i-1}O_{i-1}p_{i-1}$. 
In the context of circle $C_i$, inequality b) becomes 
$\angle w'_{i}O'_{i}p_{i} < \angle w_{i}O_{i}p_{i}$ and the second inequality
in (\ref{eq:anglesp}) follows from Lemma~\ref{lem:angles}. 
Let $r$ be the rightmost intersection of $\Cpi{i-1}$ with $st$.
It follows that $\angle w'_{i-1}O'_{i-1}r \leq 3\pi/2$. Since $p_i$ lies on
$\arcP{i-1}{w'_{i-1},r}$, the third inequality in (\ref{eq:anglesp}) holds.
Finally, the
first inequality in (\ref{eq:anglesp}) holds since either $p_i=w'_i$ or
$p_i$ lies on $\arcP{i}{w'_i,p_{i+1}}$ that is clockwise oriented.
\end{proofof}

\end{document}